\documentclass[12pt]{article}
\usepackage{amsmath,mathrsfs, eufrak, mathtools,amsthm}
\usepackage{amssymb}
\usepackage{nicefrac}
\usepackage{newtxmath,newtxtext}
\usepackage{microtype}
\usepackage{titlesec}
\usepackage[top=1.0in,bottom=1.0in,left=1.5in,right=1in]{geometry}
\usepackage{soul}

\titleformat{\subsubsection}[runin]{\normalfont\bfseries}{\thesubsubsection}{1em}{}

\usepackage{pstricks,pst-plot,psfrag}

\usepackage{graphicx,subfigure,xspace,bm,tikz,pgfplots}
\usepackage[lined,linesnumbered,algoruled,noend]{algorithm2e}

\newtheorem{theorem}{Theorem}[section]
\newtheorem{definition}{Definition}

\newtheorem{lemma}{Lemma}

\newtheorem{remark}{Remark}

\newtheorem{example}{Example}
\newtheorem{corollary}[theorem]{Corollary}
\newtheorem{proposition}[theorem]{Proposition}  
\renewcommand{\qed}{\hfill $\blacksquare$}

\newcounter{para}[section]

\renewcommand{\d}{\operatorname{d}}
\newcommand{\ad}{\operatorname{ad}}

\newcommand{\oprocendsymbol}{\hbox{$\bullet$}}
\newcommand{\oprocend}{\relax\ifmmode\else\unskip\hfill\fi\oprocendsymbol}

\newcommand{\g}{\mathfrak{g}}

\newcommand{\R}{\mathbb{R}}
\newcommand{\1}{{\bf 1}}
\newcommand{\A}{\mathbb{A}}
\newcommand{\bbL}{\mathbb{L}}
\newcommand{\K}{{\rm K}}
\renewcommand{\cal}{\mathcal}

\newcommand{\dep}{\operatorname{dep}}

\renewcommand{\sl}{\mathfrak{sl}}

\newcommand{\tr}{\operatorname{tr}}

\renewcommand{\span}{\operatorname{Span}}
\newcommand{\rp}{{\rm p}}

\definecolor{BBlue}{cmyk}{.98,0.10,0,.25}

\begin{document}

\title{\Large Controllability of Continuum Ensemble of Formation Systems over Directed Graphs}
%\rule{\textwidth}{0.4mm}}
\date{}
\maketitle

\vspace{-2cm}
\begin{flushright}
{\small {\bf Xudong Chen\footnote[1]{X. Chen is with the ECEE Dept., CU Boulder. {\em Email: xudong.chen@colorado.edu}.} 
}}
\end{flushright}

\begin{abstract}
	We propose in the paper a novel framework for using a common control input to simultaneously steer an infinite ensemble of networked control systems. We address the problem of co-designing information flow topology and network dynamics of every individual networked system so that a continuum ensemble of such systems is controllable. To keep the analysis tractable, we focus in the paper on a special class of ensembles systems, namely ensembles of multi-agent formation systems. Specifically, we consider an ensemble of formation systems indexed by a parameter in a compact, real analytic manifold. Every individual formation system in the ensemble is composed of $N$ agents. These agents evolve in $\mathbb{R}^n$ and can access relative positions of their neighbors. The information flow  topology within every individual formation system is, by convention, described by a directed graph where the vertices correspond to the $N$ agents and the directed edges indicate the information flow. For simplicity, we assume in the paper that all the individual formation systems share the same information flow  topology given by a common digraph $G$. Amongst other things, we establish a sufficient condition for approximate path-controllability of the continuum ensemble of formation systems. We show that if the digraph $G$ is strongly connected and the number $N$ of agents in each individual system is great than $(n + 1)$, then every such system in the ensemble is simultaneously approximately path-controllable over a path-connected, open dense subset.
	\end{abstract}

\section{Introduction}\label{sec:intro}

Ensemble control deals with the problem of using a single control input to simultaneously steer a large population  (and in the limit, a continuum) of dynamical systems. Consider a general ensemble of dynamical systems indexed by a parameter $\sigma$ of a certain  parameterization space $\Sigma$, which   can be either finite, countably infinite, or a locally Euclidean space.  We call an individual dynamical system in the ensemble system-$\sigma$, for $\sigma\in \Sigma$,  if it is associated with the parameter~$\sigma$. Denote by $x_\sigma(t) \in \R^l$ the state of system-$\sigma$ at time~$t$. Then, in its most general form, the dynamics of an ensemble control system can be described by the following differential equation:     
$$
\dot x_\sigma(t):= \frac{\partial}{\partial t} x_\sigma(t) = f(x_\sigma(t), \sigma, u(t)), \quad \sigma\in \Sigma,
$$  
where $u(t)\in \R^m$ is a common control input that applies to every individual system in the ensemble.

Controllability of an ensemble control system is, roughly speaking,  the ability of using the common control input $u(t)$ to simultaneously steer every individual system from any initial condition $x_\sigma(0)$ to any final condition $x_\sigma(T)$ at any given time~$T$. A precise definition of ensemble controllability will be given in Section~\ref{sec:Main}.  
Ensemble control originated from physics (e.g., NMR spectroscopy~\cite{ernst1987principles,glaser1998unitary}), and naturally has many applications across various disciplines of science and engineering. These applications include: \textit{(i)} Spintronics for spin-logic computation~\cite{behin2010proposal,appelbaum2016spin}, \textit{(ii)} smart materials that can respond to external stimuli such as light~\cite{yu2003photomechanics} and heat~\cite{taniguchi2018walking}, and \textit{(iii)}
control of neuron activities and brain dynamics~\cite{hammond2007pathological,li2013control,jadhav2016coordinated,mardinly2018precise}, just to name a few. 

{\em Ensemble of networked control systems.} We propose in the paper a novel framework that applies the idea of ensemble control to large scale multi-agent systems.  
The question of how to control a multi-agent system is not new. 
Existing approaches to the question often rely on the use of local interactions (e.g., communication and sensing) among agents, which turn the controllability problem into a problem of designing the underlying network topology that governs the information flow among the agents~\cite{baillieul2003information,rahmani2009controllability,liu2011controllability,sundaram2013structural,chen2017controllability}. But a larger networked system tends to be more fragile, less flexible, and less scalable; indeed, adding new agents into or removing agents out of the system changes its network topology, and can cause the entire system to lose controllability. Diagnosis and remediation can be very complicated especially when the system size is large. 

The ensemble control framework which we propose below provides an alternative method for controlling large scale multi-agent systems. We will consider an {\em extreme scenario} where an ensemble system is composed of {\em infinitely many} agents which are loosely connected ---- loose in the sense that the agents in the ensemble form relatively small networks and these networks do not necessarily have to interact with each other (in order that the ensemble system is controllable). 
Because every multi-agent system that is composed of a finite number of such small networks can be viewed as a proper subsystem of the  infinite ensemble, they constitute as special cases of the extreme scenario. Controllability of the infinite ensemble system will guarantee the controllability of any finite subsystem of it. As a consequence, any finite ensemble of networked systems will be flexible and scalable; adding new (or removing existing) individual networks has no impact on controllability of the others. 

To this end, we consider a continuum ensemble of dynamical systems where every  individual system is itself a networked control system composed of finitely many agents. For ease of presentation, we assume in the sequel that every individual system has the same number of agents, which we denote by~$N$.  Let $x_{\sigma,i}(t)$ be the state of agent~$i$ at time~$t$ within the individual networked system-$\sigma$, and~$h_{\sigma,i}(t)$ be the local information accessible to the agent~$i$. The collection $\{h_{\sigma,i}(t)\}^N_{i = 1}$ thus completely determines the information flow within system-$\sigma$ at time~$t$. Then, in its most general form, the dynamics of an ensemble of networked systems, compliant with the information flows, can be described by the following differential equation:          
$$
	\dot x_{\sigma,i}(t)= f_i(h_{\sigma, i}(t), \sigma, u(t)), \quad 1 \le i \le N  \mbox{ and }  \sigma\in \Sigma. 
$$
We address in the paper the controllability issue of the ensemble of networked systems:   
How to co-design the information flow and the dynamics for every individual networked system so that an ensemble of such systems is controllable?    

{\em Ensemble formation system.} To keep the above co-design problem tractable, we focus in the paper on a special class of ensemble systems, namely ensembles of multi-agent formation systems. The class of ensemble formation systems investigated here can be viewed as a prototype for the study of design and control of other general ensembles of networked control systems. We describe below in details the model of an ensemble formation system considered in the paper.  

Let the parameterization space $\Sigma$ be a compact, real analytic manifold. Each individual system-$\sigma$ in the ensemble is composed of~$N$ agents all of which evolve in an $n$-dimensional Euclidean space~$\R^n$. 

The information flow within an individual formation system is, by convention, described by a digraph. 
Specifically, let $G_\sigma = (V, E_\sigma)$ be a digraph with $V = \{v_i\}^N_{i = 1}$ the vertex set and $E_\sigma$ the edge set;  
if $v_iv_j$ is an edge of $G_\sigma$ (from~$v_i$ to~$v_j$), then agent~$i$ of system-$\sigma$ can access the relative position $(x_{\sigma, j}(t) - x_{\sigma,i}(t))$ between agent~$j$ of the same individual system and itself. 

For simplicity, we assume that the information flows of all individual formation systems are described by a {\em common} digraph $G = (V, E)$, i.e., $G_\sigma = G$ for all $\sigma\in \Sigma$.  
For a given vertex~$v_i$, we let $V^-_i:= \{v_j \mid v_iv_j \in E\}$ be the set of outgoing neighbors of~$v_i$, then the local information accessible to the agent~$i$ in system-$\sigma$ is given by $$h_{\sigma,i}(t) = \{x_{\sigma,j}(t) - x_{\sigma,i}(t) \mid v_j \in V^-_i\}.$$ 

To keep the analysis tractable, we assume that the dynamics $f_i(h_{\sigma,t}(t), \sigma, u(t))$ is separable in local information $h_{\sigma, i}(t)$,  parameter~$\sigma$, and common control input~$u(t)$. Specifically, we consider the following control model for each agent~$i$ in system-$\sigma$:   
\begin{equation}\label{eq:startensemble}
	\dot x_{\sigma,i}(t) = \sum_{v_j\in V^-_i}  \sum^r_{s = 1}u_{ij, s}(t) \rho_s(\sigma) (x_{\sigma,j}(t) - x_{\sigma,i}(t)),     
\end{equation}
where each $\rho_s: \Sigma \to \R$, for $1\le s \le r$, is a real analytic function and each $u_{ij, s}:[0,T] \to \R$, for $1\le s \le r$ and $v_i v_j \in E$,  is an integrable function over any finite time interval $[0,T]$.

We call each $\rho_s$ a {\bf parameterization function}. These parameterization functions describe the way in which individual formation systems in the ensemble differ from each other, and can be thought as the diversity of the individual systems in the ensemble. We will see that such a diversity is necessary for an ensemble system to be controllable.     

For ease of notation, we let $u(t)\in \R^{r|E|}$ be the collection of all the scalar control inputs $u_{ij, s}(t)$. 
We note again that the same control input $u(t)$ applies to every individual formation system in the ensemble.  We call system~\eqref{eq:startensemble} an {\bf ensemble formation system}, and refer to Fig.~\ref{fig:ensembleformationsys} for an illustration.

\begin{figure}[h]
\begin{center}
	\includegraphics[width = 0.6\textwidth]{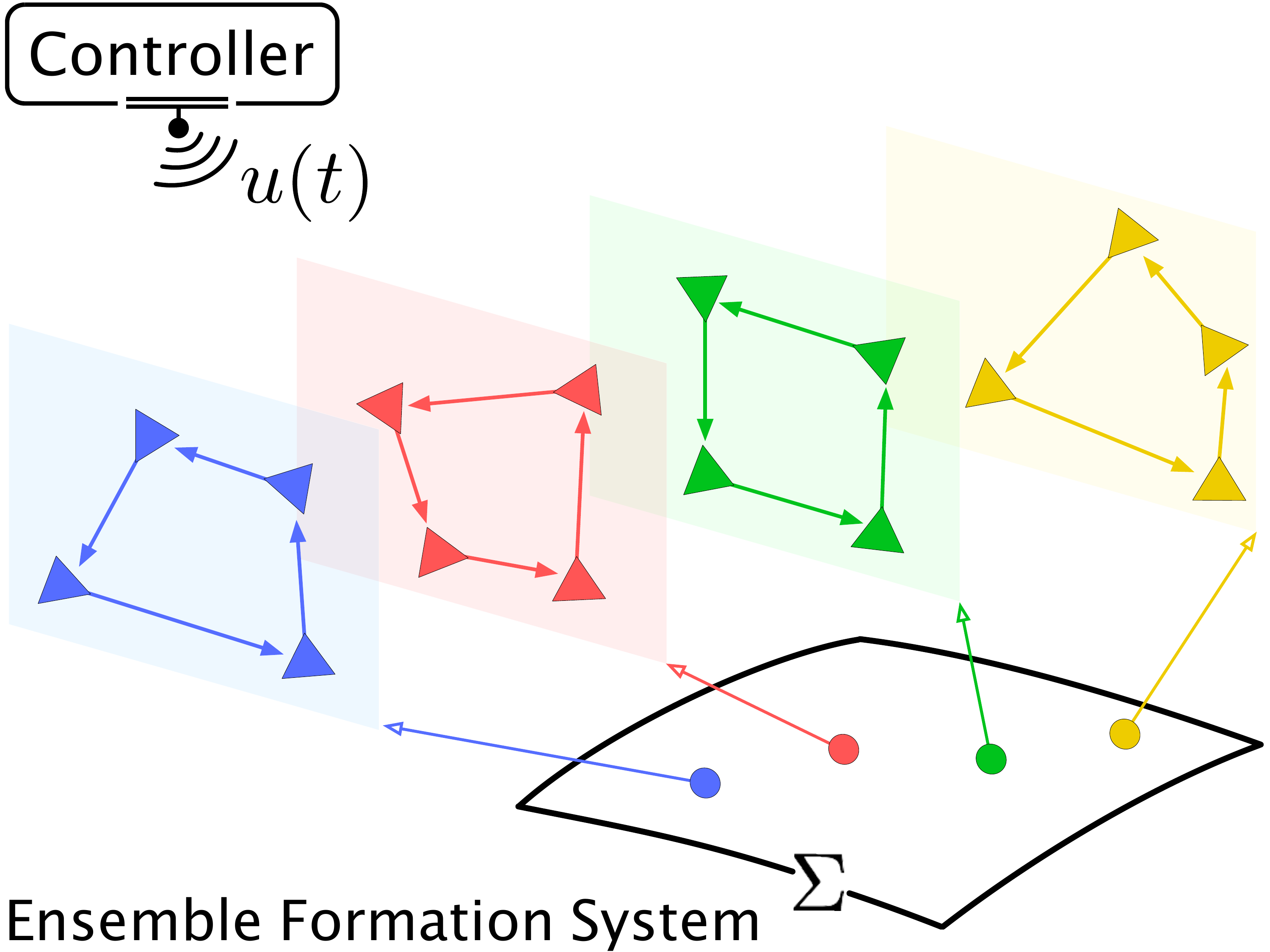}
	\caption{A continuum ensemble of formation systems is indexed by a parameter~$\sigma$ in a two dimensional surface~$\Sigma$. Each individual formation system is composed of 4 agents  evolving in $\R^2$. A cycle digraph, common to all, describes the information flow $(x_{\sigma,j}(t) - x_{\sigma, i}(t))$ within every individual formation system. A common control input $u(t)$, composed of $4$ scalar signals each of which corresponds to an edge of the cycle digraph, applies to all individual formation systems in the ensemble. The controllability result (Theorem~\ref{thm:main}) implies that such an ensemble formation system is approximately path-controllable.}\label{fig:ensembleformationsys} 
\end{center}	
\end{figure}

{\em Outline of contribution.} We establish in the paper a sufficient condition for the ensemble formation system~\eqref{eq:startensemble} to be approximately path-controllable. Roughly speaking, this is about the capability of using a common control input to simultaneously steer every individual formation system in the ensemble to approximate any desired trajectory of formations. Note, in particular, that trajectories of different individual systems can be completely different. We refer to Def.~\ref{def:ensemblecontrollability} for a precise definition, to Fig.~\ref{fig:ensemblecontrol} for an illustration, and to Theorem~\ref{thm:main} for the controllability result.  
The theorem addresses the interplay between the information flows within the individual formation systems (i.e., the common digraph~$G$), the parameterization functions $\{\rho_s\}^r_{s = 1}$, and the controllability of the ensemble formation system.

A key component of the analysis of the ensemble formation system involves computing the iterated Lie brackets of control vector fields of system~\eqref{eq:startensemble}, which further boils down to the computation of iterated matrix commutators of certain sparse zero-row-sum matrices. We provide in Section~\ref{sec:liealgebraofzrs} an in-depth analysis of such matrix commutators.   
Specifically, 
for an edge $v_iv_j$ of the digraph $G$, we let $A_{ij}$ be an $N\times N $ matrix with~$1$ on the $ij$th entry, $-1$ on the $ii$th entry, and~$0$ elsewhere.  We investigate the following iterated matrix commutators (or Lie products) of depth~$k$:
 $$
 [A_{i_0j_0}, [A_{i_1j_1}, \cdots [A_{i_{k - 1}j_{k - 1}}, A_{i_kj_k}]]], 
 $$
 where $v_{i_0}v_{j_0},\ldots, v_{i_k}v_{j_k}$ are edges of the given digraph~$G$. 
 
Amongst other things, we show that if $G$ is strongly connected, then, for any sufficiently large~$k$ ($k \ge \d(G)$ with $\d(G)$ the diameter of $G$), the vector space spanned by the above Lie products of depth~$k$ will be stabilized, given by the vector space of all zero-row-sum matrices with zero-trace.  Moreover, we show that there exists a family of spanning sets, termed semi-codistinguished sets (introduced in Definition~\ref{def:codistinguished}), of such a vector space---where each spanning set corresponds to a strongly connected digraph $G$---such that for any $k\ge \d(G)$, the matrices in any one of the spanning sets can be obtained by evaluating certain Lie products of the $A_{ij}$'s of the given depth~$k$. (see Theorem~\ref{thm:familycodistset} for details). 

The above results about iterated commutators of zero-row-sum matrices are instrumental in establishing approximate path-controllability of an ensemble formation system, or more generally, an ensemble of networked control systems whose dynamics are governed by zero-row-sum matrices (e.g., an ensemble of continuous-time Markov chains). Those results might also be of independent interest in the study of stochastic Lie algebra (i.e., the Lie algebra of zero-row-sum matrices~\cite{boukas2015structure,guerra2018stochastic}).

The ensemble controllability result as well as the analysis of stochastic Lie algebra carried out in the paper significantly extends the result and analysis in~\cite{chen2017controllability} where we established approximate path-controllability of a {\em single} formation system (i.e., for the case where $\Sigma$ is a singleton):      
\begin{equation}\label{eq:startingsystem}
\dot x_i(t) = \sum_{v_j\in V^-_i}  u_{ij}(t)(x_j(t)-x_i(t)), \quad 1\le i \le N. 
\end{equation}  
There, we have also computed the Lie brackets of control vector fields of system~\eqref{eq:startingsystem} and verified that~\eqref{eq:startingsystem} meets the Lie algebra rank condition under some mild assumption. However, the controllability of a {\em single} control-affine system is far from sufficient for a (continuum) ensemble of such systems to be controllable (regardless of what parameterization functions are used). In fact, a necessary condition for ensemble controllability of~\eqref{eq:startensemble} is such that the Lie algebra generated by the control vector fields {\em cannot} be nilpotent.  

On the stochastic Lie algebra level, a key difference between this paper and~\cite{chen2017controllability} is the following: In~\cite{chen2017controllability}, we computed the vector space spanned by Lie products of the $A_{ij}$'s of {\em all depths} while in this paper, we compute an infinite sequence of vector spaces spanned by Lie products of the $A_{ij}$'s of {\em any given depth}~$k$ for $k \ge \d(G)$.

{\em Literature review on ensemble controllability.} 
The controllability issue of a continuum ensemble of control-affine systems has recently been addressed in~\cite{agrachev2016ensemble}. The authors established an ensemble version of the Rachevsky-Chow theorem for systems of the type:  
$$\dot x_\sigma(t) = \sum^m_{i = 1}u_i(t)f_i(x_\sigma(t),  \sigma), \quad   \sigma\in \Sigma,$$ where the state $x_\sigma(t)$ of each individual system-$\sigma$ belongs to a real analytic manifold~$M$. The ensemble version of the Rachevsky-Chow theorem can be used as a sufficient condition for the above ensemble system to be approximately controllable. We briefly review such a condition below. First, recall that a Lie bracket $[f_i, f_j]$ of two vector fields $f_i$ and $f_j$ over~$M$ is defined such that for any smooth function~$\phi$ on~$M$, $[f_i, f_j]\phi = f_if_j\phi - f_jf_i\phi$ where  $f\phi$ is the Lie derivative of~$\phi$ along a vector field~$f$. For the case where $M$ is an Euclidean space, then $[f_i, f_j]$ can be simply defined by $$[f_i, f_j] = \frac{\partial f_j} {\partial x} f_i - \frac{\partial f_i} {\partial x} f_j.$$    
Then, the ensemble version of the Rachevsky-Chow theorem established in~\cite{agrachev2016ensemble} requires that for all $x\in M$, the span of the following Lie products of control vector fields: 
$$[f_{\alpha_0}, [f_{\alpha_1}, \cdots [f_{\alpha_{k - 1}}, f_{\alpha_k}]]],$$ of all depths, when evaluated at~$x$, be dense in ${\rm L}^1(\Sigma, T_xM)$, where $T_xM$ is the tangent space of~$M$ at~$x$ and ${\rm L}^1(\Sigma, T_xM)$ is the Banach space of all integrable functions $\phi: \Sigma \to T_x M$, i.e., $\int_\Sigma \|\phi(\sigma)\|d\sigma < \infty$. 

We next mention~\cite{chen2018structure} in which a special class of control-affine ensemble systems was investigated. Specifically, the dynamics of each system-$\sigma$ investigated there is separable in the state~$x_\sigma(t)$, the control~$u(t)$, and the parameter~$\sigma$: 
\begin{equation}\label{eq:review}
\dot x_\sigma(t) = \sum^m_{i = 1}\sum^r_{s = 1} u_{i,s}(t) \rho_s(\sigma) f_i(x_\sigma(t)), \quad  \sigma\in \Sigma. 
\end{equation} 
Moreover, the set of control vector fields $\{f_i\}^m_{i = 1}$ satisfies the following three conditions: {\em (i)} For all $x\in M$, $\{f_i(x)\}^m_{i = 1}$ spans $T_xM$;  {\em (ii)} For any two vector fields $f_i, f_j$, there exist $f_k$ and a constant $\lambda\in \R$ such that $[f_i, f_j] = \lambda f_k$, and conversely; {\em (iii)} For any $f_k$, there exist $f_i$, $f_j$, and a nonzero $\lambda$ such that $[f_i, f_j] = \lambda f_k$. We call any such $\{f_i\}^m_{i = 1}$ a {\em distinguished set of vector fields}, and system~\eqref{eq:review} a {\em distinguished ensemble system}. It was shown that the ensemble version of the Rachevsky-Chow theorem established in~\cite{agrachev2016ensemble} holds for a distinguished ensemble system (provided that a mild assumption on the parameterization functions is satisfied). We further refer to~\cite{JSL-NK:09} for ensemble control of Bloch equations as a motivating example of a distinguished ensemble system.           
  
We note here that the ensemble formation system~\eqref{eq:startensemble} is of type~\eqref{eq:review}, i.e., the dynamics is separable in the state,  the parameter, and the control input. But, system~\eqref{eq:startensemble} is {\em not} distinguished. Specifically, we will see that the set of control vector fields of the ensemble formation system~\eqref{eq:startensemble} satisfies conditions {\em (i)} and {\em (iii)}, but not~{\em (ii)}. Correspondingly, we will modify the arguments used in~\cite{chen2018structure} and establish ensemble controllability of system~\eqref{eq:startensemble}.

 {\em Organization of the paper.}  We introduce in Section~\ref{sec:definitions} preliminaries, key definitions and notations. We state in Section~\ref{sec:Main} the main result (Theorem~\ref{thm:main}) about controllability of the ensemble formation system~\eqref{eq:startensemble}. A sketch of the proof of the theorem will be given at the end of the section.  Next, in Section~\ref{sec:liealgebraofzrs}, we investigate the stochastic Lie algebra. We compute the iterated matrix commutators of the zero-row-sum matrices $A_{ij}$ of a given depth. The main result of the section we will establish is Theorem~\ref{thm:familycodistset}.     
 Then, in Section~\ref{sec:proofofmainresult}, we analyze the ensemble formation system and establish the controllability result. We provide conclusions at the end.             

\section{Definitions and notations}\label{sec:definitions}
We introduce here key definitions and notations.

{\em 1. Vector space.} Denote by $\{e_i\}^N_{i = 1}$ the standard basis of the Euclidean space $\R^N$, and denote by $\1\in \R^N$ the vector of all ones. 

For any vector $v$ in a Euclidean space, denote by $\|v\|$ the two-norm (i.e., the Euclidean norm). For an arbitrary matrix $X$, denote by $\|X\|$ the induced matrix two-norm.

For a subset $S$ of a vector space $V$, we denote by $\span(S)$ the subspace of~$V$ spanned by the elements in $S$. The {\bf negative} of $S$, denoted by $-S$, is composed of all vectors $-v$ for $v\in S$. We further let $\pm S$ be the union of $S$ and $-S$. 
 
For two subspaces $V'$ and $V''$ of $V$, we let $V' + V''$ be the subspace of $V$ spanned by all vectors $v' + v''$ for $v'\in V'$ and $v''\in V''$. 

{\em 2. Lie algebra.} Let $\g$ be a real Lie algebra, with $[\cdot, \cdot]$ the Lie bracket. For $\g'$ and $\g''$ two subspaces of $\g$, we let $[\g',\g'']$ be defined as the span of $[X', X'']$ for $X'\in \g'$ and $X''\in \g''$. We will use such a notation to define the commutator ideal $[\g, \g]$ of $\g$, and more generally, the lower central series $[\g,\cdots, [\g, \g]]$.  

However, if $S'$ and $S''$ are two {\em finite} subsets of $\g$, we let $[S',S'']$ be a {\em finite} subset of~$\g$ composed of all $[X', X'']$ for $X'\in \g'$ and $X''\in \g''$.       
So, for example, using the above notation, we have $$[\span(S'), \span(S'')] = \span([S', S'']).$$ 

Denote by $\ad$ the {\bf adjoint action}, i.e., for any $X'\in \g$, we let $\ad(X'): \g\to \g$ be defined as $\ad(X')(X''):= [X', X'']$ for any $X''\in \g$.  
For any two finite subsets $S'$ and $S''$ of $\g$, we let $\ad(S')(S'') := [S', S'']$. Next, we define via recursion a sequence of finite subsets of~$\g$ as follows: For $k = 0$, we define $\ad^0(S')(S'') := S''$; for $k \ge 1$, we define  
$$\ad^k(S')(S'') := [S', \ad^{k-1}(S')(S'')].$$ 
 Further, if $S' = S''$, then, for ease of notation, we simply write $\ad^k(S'):= \ad^k(S')(S')$. 

Let $S:= \{A_1,\ldots, A_k\}$ be an arbitrary set. Denote by $\mathfrak{L}(S)$ the free Lie algebra generated by the $A_i$'s treated as the free generators. For a Lie product $A\in \mathfrak{L}(S)$, let $\dep(A)$ be the {\bf depth} of~$A$ defined as the number of Lie brackets in~$A$. Equivalently, $\dep(A)$ is the number 
of~$A_i$'s in~$A$ (counted with multiplicity) minus one. For example, the depth of $[A_{i_1}, [A_{i_2}, A_{i_3}]]$ is~$2$. 
Let ${\cal S}\subset \mathfrak{L}(S)$ be the collection of Lie products. 
We further decompose ${\cal S} = \sqcup_{k \ge 0} {\cal S}(k)$ where each ${\cal S}(k)$ is composed of Lie products of depth~$k$. 

{\em 3. Directed graph.} 
Let $G = (V,E)$ be a directed graph (or simply, {\em digraph}) of~$N$ vertices with $V = \{v_i\}^N_{i= 1}$ the set of vertices and $E$ the set of edges. We assume in the paper that a digraph $G$ does {\em not} have any self-loop. 
Denote by $v_iv_j$ a directed edge of $G$ from~$v_i$ to~$v_j$. We call $v_j$ an {\bf out-neighbor} of~$v_i$, and denote by $V^-_i$ the set of out-neighbors of~$v_i$.  

Let $v_{i_0}\cdots v_{i_l}$ be a {\bf path} where each $v_{i_{p-1}}v_{i_p}$, for $p = 1,\ldots, l$, is an edge of $G$. Note that the vertices $\{v_{i_j}\}^{l - 1}_{j = 0}$ have to be pairwise distinct. If $v_{i_0} = v_{i_l}$, then the path is a {\bf cycle}. The {\bf length} of a path (cycle) is the number~$l$ of edges in it.
  
We call $G$ {\bf weakly connected} if the undirected graph obtained by ignoring the orientation of the edges is connected. 
The digraph $G$ is {\bf strongly connected} if for any pair of distinct vertices $v_i$ and $v_j$, there is a path from $v_i$ to $v_j$. 

The {\bf diameter} of a strongly connected digraph~$G$, denoted by~$\d(G)$, is the smallest positive integer number such that the following hold: For any two vertices $v_i$ and $v_j$ (possibly the same),  there exists a path of length~$l$ with $l \le \d(G)$ from~$v_i$ to~$v_j$. For example, the diameter of a cycle digraph of~$n$ vertices is~$n$. 

Given a subset $V'$ of $V$, a subgraph $G' = (V',E')$ is {\bf induced by} $V'$ if its edge set~$E'$ contains all edges in~$E$ that connect vertices in~$V'$, i.e., $E' := \{v_i v_j\in E \mid v_i, v_j \in V'\}$.

{\em 4. Algebra of functions.} Let $\Sigma$ be an arbitrary space. 
Given a real-valued function $\rho$ on $\Sigma$ and a nonnegative integer~$k$, we define $\rho^k$ as $\rho^k(\sigma):= \rho(\sigma)^k$ for all $\sigma \in \Sigma$. Note, in particular, that if $k = 0$, then $\rho^0$ is a constant function on $\Sigma$ whose value is~$1$ everywhere. 
We say that~$\rho$ is {\bf everywhere nonzero} if $\rho(\sigma) \neq 0$ for all~$\sigma\in \Sigma$. Note that for any such function, $\rho^{-1}$ is well defined, given by $\rho^{-1}(\sigma) := \rho(\sigma)^{-1}$ for all $\sigma\in \Sigma$. Similarly, we define $\rho^{-k}:= (\rho^{-1})^k$ for any~$k \ge 0$.     
	
Let $\{\rho_s\}^r_{s = 1}$ be a set of functions over $\Sigma$. We call a function $\prod^r_{s = 1}\rho^{k_s}_s$ for $k_s \ge 0$ a {\bf monomial}. The {\bf degree} of the monomial is given by $\sum^r_{s = 1}k_s$. Denote by ${\cal P}$ the collection of all monomials. We decompose ${\cal P}$ as ${\cal P} = \sqcup_{k \ge 0} {\cal P}(k)$, where ${\cal P}(k)$ is the collection of all monomials of degree~$k$. 

The set of functions $\{\rho_s\}^r_{s = 1}$ is said to {\bf separate points} if for any two distinct points $\sigma$ and $\sigma'$ in $\Sigma$, there exists a function $\rho_s$ out of the set such that $\rho_s(\sigma) \neq \rho_s(\sigma')$. 

{\em 5. Control system.} For a general control system $\dot x(t) = f(x(t), u(t))$,  we denote by $u[0,T]$ the control input~$u(t)$ over the time interval $[0,T]$, and  $x[0,T]$ the trajectory~$x(t)$ over the interval $[0,T]$ generated by the control input. 

\section{Controllability of formation systems}\label{sec:Main}
\subsection{Controllability of a single formation system}\label{ssec:resultsingle}
We review in the subsection the controllability result established in~\cite{chen2017controllability} for a {\em single} formation system.  
The formation control system considered there is composed of $N$ agents $x_1(t),\ldots, x_N(t)$ that evolve in $\R^n$.  We use, by convention, a directed graph $G = (V,E)$ to indicate the information flow among the~$N$ agents: If $v_iv_j$ is an edge of $G$, then agent~$i$ can access the relative position $(x_j(t) - x_i(t))$ between agents~$i$ and~$j$. We assume that the dynamics of each agent~$i$ at any time~$t$ is given by a certain linear combination of $(x_j(t) - x_i(t))$ for $v_j\in V^-_i$.   
Then, the way a controller steers an agent~$i$ is to manipulate the coefficients associated with the linear combination. Specifically, we have the following dynamics for each agent~$i$: 
\begin{equation}\label{eq:singleformationsys}
	\dot x_i(t) = \sum_{v_j\in V^-_i}  u_{ij}(t)(x_j(t)-x_i(t)), \quad 1\le i \le N, 
\end{equation}
where each $u_{ij}(t)$ is a scalar control input. Note that~\eqref{eq:singleformationsys} is a bilinear control system, i.e., the dynamics is linear in the state and the control input. 
We also note that the above control dynamics can be viewed as a variation of the classical diffusively-coupled dynamics $\dot x_i = \sum_{v_j}a_{ij} (x_j - x_i)$ for which one replaces the (positive) coefficients $a_{ij}$ with the control inputs $u_{ij}(t)$.

The dynamics of the above formation system can be written into a matrix form. For that, we need to introduce a few definitions and notations. We first have the following one:

\begin{definition}[Primary matrix]\label{def:primarymatrix}
For a digraph $G = (V, E)$ of $N$ vertices, we define for each edge $v_iv_j$ of~$G$ a zero-row-sum $N\times N$ matrix as follows:
\begin{equation}\label{eq:defAij}
	A_{ij}:= e_ie_j^\top - e_ie_i^\top ,
\end{equation}
i.e., $A_{ij}$ has~$1$ on the $ij$th entry, $-1$ on the $ii$th entry, and~$0$ elsewhere. We call any such matrix $A_{ij}$ a {\bf primary matrix}. 
\end{definition}

Next, for any given time~$t$, let $X(t) := [x^\top_1(t); \cdots;  x^\top_N(t)]$ be an $N\times n$ matrix, i.e., the $i$th row of~$X(t)$ is $x^\top_i(t)$. 
We call $X(t)$ a {\em configuration}, and denote by $P:= \R^{N \times n}$ the {\em configuration space}. 
With the above notations, we can re-write system~\eqref{eq:singleformationsys} into the following differential equation for the matrix $X(t)$:  
\begin{equation}\label{eq:rewrittensinglesystem}
\dot X(t) = \sum_{v_iv_j \in E}u_{ij}(t)A_{ij}X(t).
\end{equation}   
We have investigated in~\cite{chen2017controllability} approximate path-controllability of system~\eqref{eq:rewrittenensemblesystem}. Roughly speaking, 
a control system is approximately path-controllable if one is able to steer the system to approximate any target trajectory of states. A precise definition is given below.

\begin{definition}\label{def:singlecontrollability} 
Let $Q$ be an open, path-connected subset of~$P$. System~\eqref{eq:rewrittensinglesystem} is {\bf approximately path-controllable} over $Q$ if for any $T > 0$, any smooth trajectory $\hat X : [0,T] \to Q$, and any error tolerance $\epsilon > 0$, there are integrable functions $u_{ij}:[0,T]\to \R$, for $v_iv_j\in E$, as control inputs such that the trajectory $X[0,T]$ generated by~\eqref{eq:rewrittensinglesystem}, from an initial condition $X(0)$ with $X(0)\in Q$ and $\|X(0) - \hat X(0)\| < \epsilon$, satisfies $$\|X(t) - \hat X(t)\| < \epsilon, \quad \forall t \in  [0,T].$$ 
\end{definition}

 We illustrate the above definition in Fig.~\ref{fig:pathcontrollability}:

\begin{figure}[h]
\begin{center}
	\includegraphics[width = 0.5\textwidth]{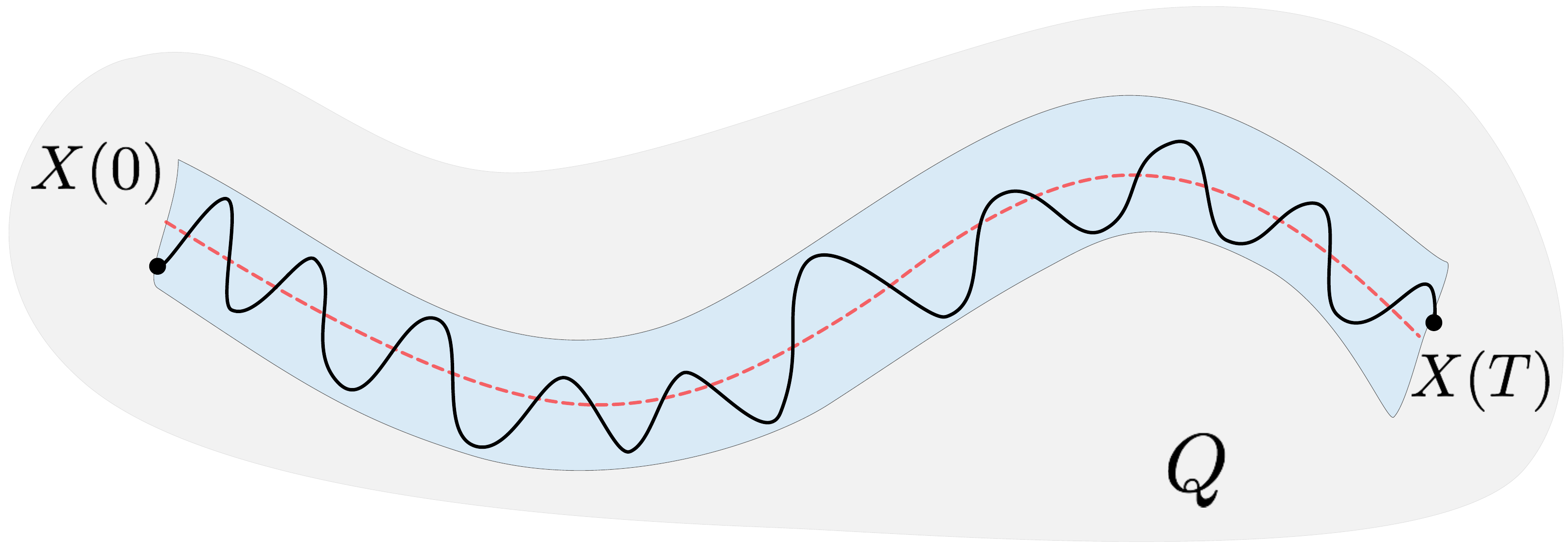}
	\caption{In the figure, the red dashed curve is a desired trajectory $\hat X[0,T]$ we want the system~\eqref{eq:rewrittensinglesystem} to follow. Each $X(t)$ belongs to a certain open, path connected subset~$Q$ of~$P$. The blue region is an  $\epsilon$-tubular neighborhood of this trajectory. The solid curve is generated by a control law $u[0,T]$ such that the solution $X[0,T]$ is within the $\epsilon$-tubular neighborhood.}\label{fig:pathcontrollability} 
\end{center}
\end{figure}

We state below the controllability result for system~\eqref{eq:rewrittensinglesystem}. To proceed, we first specify the open, path-connected set~$Q$ considered in~\cite{chen2017controllability}. We need the following definition: 

\begin{definition}\label{def:nondegenerate}
 A configuration $X = [x^\top_1;\cdots; x^\top_N]\in P$ is {\bf nondegenerate} in $\R^n$ if the span of $\{x_i- x_1,\ldots ,x_i - x_N\}$ is $\R^n$ for some (and hence, any) $i = 1,\ldots,N$. If $N = n + 1$, then~$X$ is an {\bf $n$-simplex}.   
\end{definition}

\begin{remark}\label{rmk:nondegenerate}{\em 
A configuration~$X$ is nondegenerate if and only if there does {\em not} exist a proper subspace of $\R^n$ that contains all the $x_i$'s.  For example, a line configuration is degenerate in~$\R^2$ and a planar configuration is degenerate in~$\R^3$.   
Note that if a configuration~$X$ is nondegenerate, then there exists a subset of $(n + 1)$ agents such that the sub-configuration formed by these $(n + 1)$ agents is an $n$-simplex (a nondegenerate triangle  for $n = 2$ or a nondegenerate tetrahedron for $n = 3$).} %.}
\end{remark} 

Now, let $Q$ be the collection of all  nondegenerate configurations in~$P$:
\begin{equation}\label{eq:defsetQ}
Q := \{X\in P \mid X \mbox{ is nondegenerate}\}.
\end{equation} 
We have the following fact:

\begin{lemma}\label{lem:resultforsetQ}
If $N > (n + 1)$, then $Q$ is a path-connected, open dense subset of $P$. Moreover, if the underlying digraph $G$ is strongly connected, then system~\eqref{eq:rewrittensinglesystem} is approximately path-controllable over the set~$Q$.   
\end{lemma}

A complete proof of the above result can be found in~\cite{chen2017controllability} where we have established the controllability result for a broader class of (weakly connected) digraphs.

\subsection{Controllability of an ensemble formation system}\label{ssec:resultensemble}\label{ssec:ensemble}
We now return to the ensemble formation system introduced in Section~\ref{sec:intro}. We will state in the subsection the path-controllability result for the ensemble system which straightforwardly generalizes Lemma~\ref{lem:resultforsetQ}.

Recall that if an individual formation system is indexed by $\sigma\in \Sigma$, then we call it system-$\sigma$. Each individual formation system is composed of $N$ agents in $\R^n$. We reproduce below the dynamics of agent~$i$ associated with system-$\sigma$:
\begin{equation}\label{eq:ensembleformationsys}
	\dot x_{\sigma,i}(t) = \sum_{v_j\in V^-_i}  \sum^r_{s = 1}u_{ij, s}(t) \rho_s(\sigma) (x_{\sigma,j}(t) - x_{\sigma,i}(t)). 
\end{equation}
We note again that the control inputs $u_{ij,s}(t)$ are the same for every individual formation system. 
Similarly, one can re-write the above dynamics into a matrix form as we did in the previous subsection: For any given time~$t$, we let $X_\sigma(t)$ be an $N\times n$ matrix defined as follows:
$$X_\sigma(t): = [x^\top_{\sigma,1}(t); \cdots; x^\top_{\sigma, N}(t)].$$
We call $X_\sigma(t)$ a configuration of system-$\sigma$.  
Then, by~\eqref{eq:ensembleformationsys}, we have the following differential equation for $X_\sigma(t)$: 
\begin{equation}\label{eq:rewrittenensemblesystem}
\dot X_\sigma(t) = \sum_{v_iv_j \in E}\sum^r_{s = 1} u_{ij, s}(t)\rho_s(\sigma)A_{ij}X_{\sigma}(t), \quad  \sigma\in \Sigma,  
\end{equation}
where each $A_{ij}$ is a primary matrix introduced in Def.~\ref{def:primarymatrix}.

We will now generalize approximate path-controllability for a single formation system (Def.~\ref{def:singlecontrollability}) to the ensemble case. Roughly speaking, an ensemble system is said to be approximately path-controllable if one is able to use to common control inputs to steer simultaneously every individual system to approximate any given target trajectory (different individual systems can have different target trajectories). We make the statement precise below.     
Let $$X_\Sigma(t): = \{X_\sigma(t)\mid \sigma\in \Sigma\}$$ 
be the collection of configurations at time~$t$. We call $X_\Sigma(t)$ a
{\bf profile} of the ensemble formation system. For a fix time~$t$,  we say that $X_\Sigma(t)$ is smooth if the map $\sigma \in \Sigma\mapsto X_\sigma(t) \in P$ is smooth. Further, let $X_\Sigma[0,T]:= \{X_\sigma[0,T]\mid \sigma\in \Sigma\}$ be the trajectory of $X_\Sigma(t)$ over~$[0,T]$. Similarly, we say that $X_\Sigma[0,T]$ is smooth if the map $$(t, \sigma)\in [0,T]\times \Sigma \mapsto X_\sigma(t)\in P$$ is smooth.   
We now introduce the definition about approximate path-controllability for an ensemble formation system:

\begin{definition}\label{def:ensemblecontrollability}
Let $Q$ be an open, path-connected subset of $P$. System~\eqref{eq:rewrittenensemblesystem} is {\bf approximately ensemble path-controllable} over~$Q$ if for any smooth target trajectory $\hat X_\Sigma[0,T]$, with  $\hat X_\sigma(t)\in Q$ for all $(t,\sigma)\in  [0,T]\times \Sigma$, and any error tolerance~$\epsilon$, there are integrable functions $u_{ij, s}: [0,T]\to \R$ as control inputs such that the trajectory $X_\Sigma[0,T]$ generated by~\eqref{eq:rewrittenensemblesystem}, from an initial condition~$X_\Sigma(0)$ with $X_\sigma(0)\in Q$ and $\|X_\sigma(0) - \hat  X_\sigma(0)\|< \epsilon$ for all $\sigma\in \Sigma$, satisfies
$$
\|X_\sigma(t)- \hat X_\sigma(t) \| < \epsilon,  \quad \forall (t,\sigma)\in [0,T]\times \Sigma.$$ 
\end{definition}

We illustrate the above definition in  Fig.~\ref{fig:ensemblecontrol}:

\begin{figure}[h]
\begin{center}
	\includegraphics[width = 0.5\textwidth]{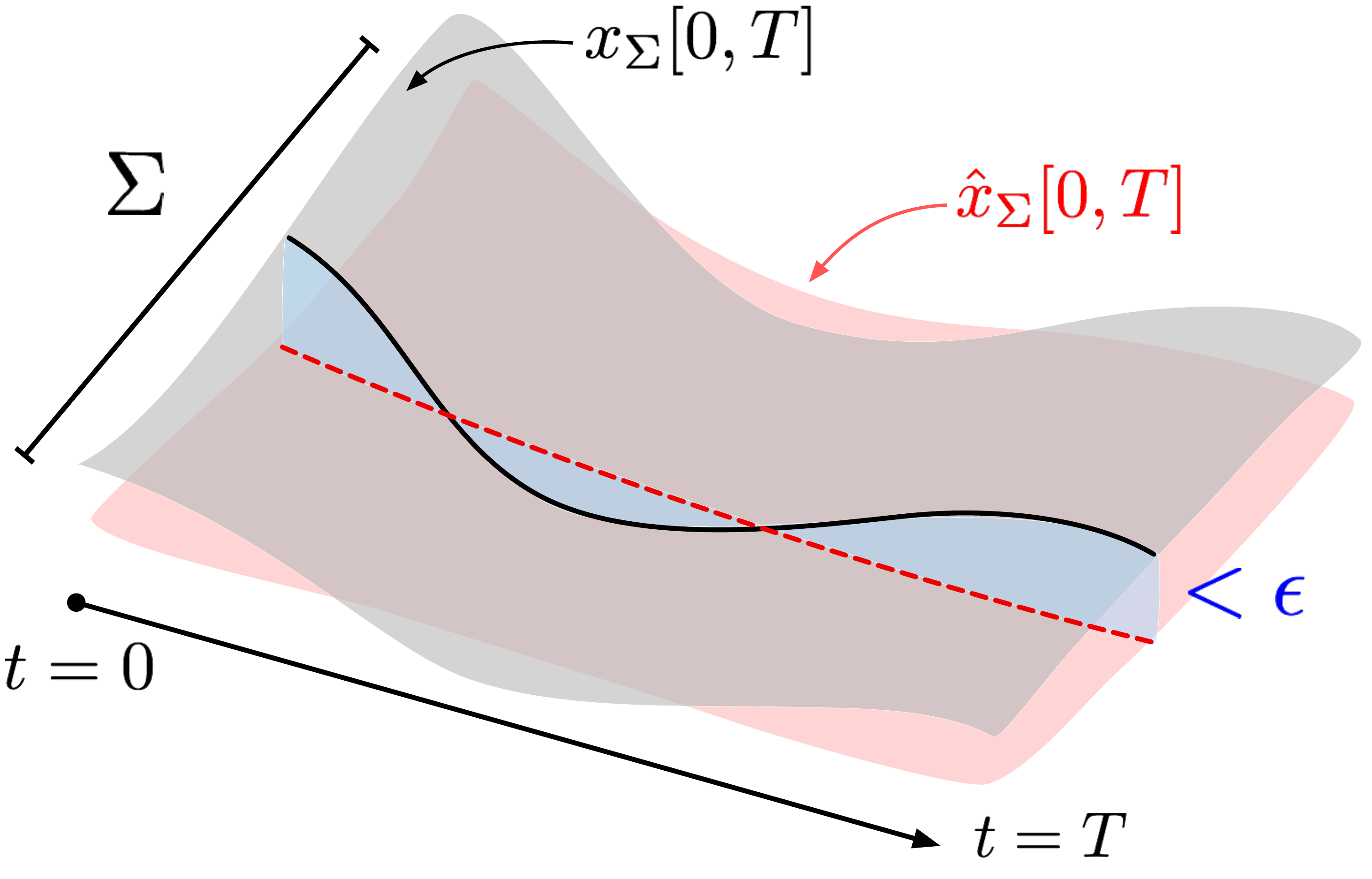}
	\caption{The red surface is a desired trajectory $\hat X_\Sigma[0,T]$ we want the ensemble formation system~\eqref{eq:rewrittenensemblesystem} to follow. Each $X_\sigma(t)$, for $(t,\sigma)\in [0,T]\times \Sigma$, belongs to an open, path connected set~$Q$.   The red dashed curve is a desired trajectory for the individual system-$\sigma$. The grey surface is the trajectory $X_\Sigma[0,T]$ generated by a common control input $u[0,T]$ such that it is within the $\epsilon$-tubular neighborhood of $\hat X_\Sigma[0,T]$. The black solid curve is the trajectory $X_\sigma[0,T]$ for the individual system-$\sigma$.}\label{fig:ensemblecontrol}
\end{center}	
\end{figure}

With the above preliminaries, we are now in a position to state the first main result of the paper (compared to Lemma~\ref{lem:resultforsetQ}):

\begin{theorem}\label{thm:main}
	Let $G$ be strongly connected and $N > (n + 1)$. Suppose that the set of parameterization functions $\{\rho_s\}^r_{s = 1}$ separates points and contains an everywhere nonzero function;  then,  system~\eqref{eq:rewrittenensemblesystem} is approximately ensemble path-controllable over the set $Q$ of nondegenerate configurations.     
\end{theorem}

A sketch of proof will be given at the end of the section. Detailed analysis will be provided in Sections~\ref{sec:liealgebraofzrs} and~\ref{sec:proofofmainresult}.   
With a few more efforts, we can extend the above result to a time-varying digraph. We recall from~\cite{chen2017controllability} the following definition:

\begin{definition}\label{def:timevaryinggraph}
A time-varying digraph $G(t)$ is {\bf right-continuous} if for any time~$t$, there exists a time duration $\delta_t > 0$ such that $G(t') = G(t)$ for all $t'\in [t, t+\delta_t)$. We call an instant $t_i$ a switching time if $\lim_{t\to t_i -}G(t) \neq G(t_i)$. 
\end{definition}
   
We now assume that the information flow of every individual formation system in the ensemble is described by a common time-varying digraph $G(t)$. With the above definition, we state the following fact as a corollary to Theorem~\ref{thm:main}:

\begin{corollary}
	Let $G(t)$ be a right-continuous time-varying digraph such that for any finite time interval, $G(t)$ has a finite number of switching times. Suppose that for any $t \ge 0$, $G(t)$ is strongly connected with $N > (n + 1)$; then,  system~\eqref{eq:rewrittenensemblesystem} is approximately ensemble path-controllable over the set~$Q$ of nondegenerate configurations.
\end{corollary}

The proof of the corollary is similar to the proof of Corollary~1 in~\cite{chen2017controllability}. For completeness of presentation, we provide below a relatively short proof of the result: 

\begin{proof}
Let $\hat X_\Sigma[0,T]$  be a desired trajectory with $\hat X_\sigma(t)\in Q$ for all $(t,\sigma) \in [0,T]\times \Sigma$. Let $t_1, \ldots, t_m \in (0,T)$ be the switching times of $G(t)$. We construct an admissible $u[0,T]$ as follows. Given a graph $G(0)$, we know from Theorem~\ref{thm:main} that there exists $u_1[0,T]$ such that system~\eqref{eq:rewrittenensemblesystem} approximates $\hat X_\Sigma$ over $[0,T]$. We use this control until the first switching time:  $u[0,t_1) := u_1[0,t_1)$. It  follows that   $\|X_\sigma(t_1)-\hat   X_\sigma(t_1)\|<\epsilon$ for all $\sigma\in \Sigma$. We can thus apply Theorem~\ref{thm:main}, but now with graph $G(t_1)$, to obtain a control law $u_2[t_1,T]$ that steers the ensemble formation system from $X_\Sigma(t_1)$ along a trajectory $X_\Sigma(t)$ such that $\|X_\sigma(t)-\hat   X_\sigma(t)\|<\epsilon$ for all $(t, \sigma)\in [t_1, T] \in \Sigma$. As before, we let $u[t_1,t_2):=u_2[t_1,t_2)$. Note that implementing the control $u(t)$ over the time interval $[0,t_2)$ yields a trajectory  $X_\Sigma[0, t_2)$ within the $\epsilon$ tolerance of $\hat   X_\Sigma[0,t_2)$ over that interval. Repeating this procedure for a finite number of times yields a control input $u[0,T]$ that can steer the ensemble formation system to approximate $\hat   X_\Sigma[0,T]$ as required.	
\end{proof}

\begin{remark}\label{rmk:conjecture}{\em 
We note here that a more general case is to assume that the underlying digraphs $G_\sigma(t)$, for $\sigma\in \Sigma$, of the individual formation systems in the ensemble are heterogeneous. Specifically, we assume that there exists a finite static graph~$G$ such that for each $\sigma\in \Sigma$ and each time~$t$, $G_\sigma(t)$ is a subgraph of $G$. Denote by $G_\Sigma(t):= \{G_\sigma(t) \mid \sigma\in \Sigma\}$ the collection of the digraphs. We call $G_\Sigma(t)$ a {\em time-varying ensemble digraph}. Def.~\ref{def:timevaryinggraph} can be transposed here by replacing $G(\cdot)$ with $G_\Sigma(\cdot)$ in the definition. We defer to another occasion the analysis of an ensemble formation system defined over a (time-varying) ensemble digraph.     
}
\end{remark}

{\em Sketch of Proof.} The proof of Theorem~\ref{thm:main} relies on the use of the so-called ``Lie extension'' of system~\eqref{eq:rewrittenensemblesystem}. We will review such a technique in Section~\ref{ssec:lieextendedsys}. By repeatedly applying the technique of Lie extension,  one arrives at the following system (with a few details omitted): 
\begin{equation}\label{eq:lieextendedsystem1234}
	\dot X_\sigma(t) = \sum u_{\alpha}(t)\rho_s(\sigma)A_{ij}X_{\sigma}(t) +  
	\sum u_{\beta}(t)\rho_s(\sigma)\rho_{s'}(\sigma)[A_{ij}, A_{i'j'}] X_\sigma(t) + \cdots 
	\end{equation}       
	Truncation after the term that involves Lie products of depth~$k$ gives rises to the $k$th order Lie extended system.  
It is known that the original system~\eqref{eq:rewrittenensemblesystem} is approximately ensemble path-controllable if and only if one (and hence any) of its Lie extended system is. It thus suffices to establish controllability of Lie extended systems. The proof is composed of two key components as we outline below:   
    \begin{itemize}
    \item {\em Commutators of primary matrices.} The control vector fields in the above Lie extension~\eqref{eq:lieextendedsystem1234} involve iterated matrix commutators of primary matrices $A_{ij}$ for $v_iv_j$ an edge of the digraph~$G$. To evaluate those control vector fields, we compute explicitly the associated matrix commutators. This is done in Section~\ref{sec:liealgebraofzrs}. Specifically, we establish in the section the following fact (Theorem~\ref{thm:familycodistset}): Let 
    $$\A^*:= \left \{A\in \R^{N\times N} \mid A\1 = 0 \mbox{ and } \tr(A) = 0\right \};$$
    then, there exists a basis $\{A^*_i\}$ of $\A^*$ such that for any sufficiently large~$k$ ($k \ge \d(G)$), the matrices $A^*_i$ can be obtained as the  matrix commutators of $A_{ij}$'s of the given depth~$k$. Consequently, system~\eqref{eq:lieextendedsystem1234} can then be simplified as follows (details will be provided in Section~\ref{ssec:lieextendedsys}):
    \begin{equation}\label{eq:thelastsystemprepre}
 	\dot X_\sigma(t) = \sum_{\rp,i} u_{i, \rp}(t) \rp(\sigma)  	 A^*_iX_\sigma(t),
 	\end{equation} 
 	where each $\rp$ is a monomial of the functions~$\{\rho_s\}^r_{s = 1}$.  \vspace{.1cm} 
 	
    \item {\em Span of control vector fields.} We analyze system~\eqref{eq:thelastsystemprepre} in Section~\ref{sec:proofofmainresult}. We show that the control vector fields in~\eqref{eq:thelastsystemprepre} satisfy the ensemble version of the Lie algebraic rank condition. More specifically, we establish in the section two facts for system~\eqref{eq:thelastsystemprepre}---one is about the span of $\{A_i^*X\}$ while the other is about function approximation by the summation $\sum_{\rp,i} u_{i, \rp}(t)\rp(\sigma)$. Specifically, we establish the following two facts: {\em (i)} The span of $\{A^*_iX\}$ is $\R^{N\times n}$ for any nondegenerate configuration~$X$ provided that $N > (n + 1)$.   This is done in Prop.~\ref{prop:fullrankLX}, Section~\ref{ssec:linearspanofcontrolvecfield}; {\em (ii)} Every continuous function $c_i(t,\sigma)$ (continuous in both arguments) can be approximated arbitrarily well by a  finite sum $\sum_{\rp,i}u_{i, \rp}(t)\rp(\sigma)$. This is essentially an application of the Stone-Weierstrass theorem. 	
    \end{itemize}

\section{Stochastic Lie algebra and semi-codistinguished sets}\label{sec:liealgebraofzrs} 
Let $\A\subseteq \R^{N\times N}$ be the vector space of all zero-row-sum (\textit{zrs}) matrices, i.e., 
$$\A := \{A \in \R^{N\times N}\mid A \1 = 0\}.$$ Denote by $[\cdot, \cdot]$ the matrix commutator, i.e., for any two $N\times N$ matrices $A_1, A_2$, we have 
$[A_1, A_2]:= A_1A_2 - A_2 A_1$.   
It should be clear that $\A$ is a Lie algebra with the matrix commutator being the Lie bracket; indeed, if $A_1\1 = A_2 \1 = 0$, then 
$
[A_1, A_2] \1= 0
$. We call~$\A$ the {\bf stochastic Lie algebra}. 

Denote by $\tr(\cdot)$ the trace of a square matrix. For any two $N\times N$ matrices $A_1$ and $A_2$, we have $\tr([A_1, A_2]) = 0$.  Define a proper subspace $\A^*$ of $\A$ as follows:
\begin{equation}\label{eq:commutatorideal}
	\A^* := \{ A\in \A \mid \tr(A) = 0 \}.
\end{equation}
If we let $Z:=I - \1\1^\top/N$ and $\R Z$ be the one-dimensional subspace of $\A$ spanned by $Z$, then $\A = \A^*\oplus \R Z$. The codimension of $\A^*$ in~$\A$ is thus~$1$.  
Recall that the {\bf commutator ideal} of the Lie algebra $\A$ is defined by $[\A, \A]$, i.e., it is the linear span of all matrix commutators $[A, A']$ for $A, A' \in\A$. By computation (see, for example,~\cite{boukas2015structure}), 
\begin{equation}\label{eq:commutatoridealcomp}
[\A,\A] = \A^*. 
\end{equation}
We need the following definition: 

\begin{definition}
	 A Lie algebra $\g$ is {\bf perfect} if $\g = [\g,\g]$.
\end{definition}

By~\eqref{eq:commutatoridealcomp}, the stochastic Lie algebra~$\A$ is {\em not} perfect. Nevertheless, its commutator ideal is perfect:  

\begin{lemma}\label{lem:perfectliealgebra}
The Lie algebra~$\A^*$ is perfect. Let $\A^* = \A^*_\mathfrak{l} \oplus \A^*_\mathfrak{r}$ be the Levi decomposition where $\A^*_\mathfrak{l}$ is semi-simple and $\A^*_\mathfrak{r}$ is the radical of $\A^*$. Then, 
$$
\A^*_\mathfrak{l} = \{A\in \A^* \mid A^\top \1 = 0\} \quad \mbox{and} \quad \A^*_\mathfrak{r} = \{\1 v^\top \mid v^\top \1 = 0\}. 
$$
Moreover, $\A^*_\mathfrak{l}$ is isomorphic to the special linear Lie algebra $\sl_{N-1}(\R) := \{ M\in \R^{(N - 1)\times (N - 1)} \mid \tr(M) = 0 \}$.
\end{lemma}

The above fact has certainly been observed in the literature~\cite{boukas2015structure,guerra2018stochastic}.  
For completeness of presentation, we provide a proof in the Appendix.

The {\bf lower central series} $\{\A_m\}_{m \ge 0}$ of $\A$ can be defined by the recursion: $\A_0:=\A$ and $\A_{m + 1} = [\A,\A_m]$ for all $m \ge 0$. It follows from Lemma~\ref{lem:perfectliealgebra} that $\A_m = \A^*$ for all $m \ge 1$.        
Recall that a primary matrix $A_{ij}$ is given by $A_{ij} = e_i e^\top_j - e_i e^\top_i$.     
For the digraph $G$, we let $S_G$ be the collection of all primary matrices $A_{ij}$ such that $v_iv_j$ is an edge of~$G$, i.e., 
$$
S_G:= \{A_{ij} \mid v_iv_j \in E\}. 
$$
We also recall that the sets $\{\ad^m(S_G)\}_{m \ge 0}$ are also defined by the recursion: $\ad^0(S_G) = S_G$ and $\ad^{m + 1}(S_G) = [S_G, \ad^{m}(S_G)]$ for all 
$m \ge 0$. It should be clear that $\ad^m(S_G) \subset \A_m $. In particular, $\ad^m(S_G) \subset \A^*$ for all $m \ge 1$. 

We now state the main result of the section. First, for the given digraph $G$, we let $S^*_G$ be a finite subset of $\A^*$ defined as follows: 
\begin{equation}\label{eq:defsG}
S^*_G:= \{A_{jk} - A_{kj},  A_{ik} - A_{ij} \mid v_j v_k \in E\}.
\end{equation} 
Note that each matrix in the set $S^*_G$ can be obtained as a commutator of certain primary matrices (see \cite{chen2017controllability,costello2014degree}):    
    \begin{equation}\label{eq:basiccommutator}
    \begin{array}{ll}
    {[A_{ij}, A_{ji} ]} = A_{ji} - A_{ij}, & 1\le i\neq j \le N, \\
    {[A_{ij}, A_{jk} ]} = A_{ik} - A_{ij}, & 1\le i\neq j \neq k \le N, \\
    {[A_{ij}, A_{ik}]} = A_{ij} - A_{ik}, & 1\le i\neq j \neq k \le N. 
    \end{array}
    \end{equation} 
However, the matrices on the left hand side of~\eqref{eq:basiccommutator} do not necessarily belong to~$S_G$ and, hence, $S^*_G$ may not belong to $[S_G, S_G]$.  Nevertheless, we will show that if $m$ is sufficiently large, then $S^*_G \subset \ad^m(S_G)$. Said in another way, for sufficiently large~$m$, every matrix in $S^*_G$ can be obtained as a certain iterated matrix commutator of the $A_{ij}$'s in $S_G$ of the given depth~$m$. 
Precisely, we have the following fact: 

\begin{theorem}\label{thm:familycodistset}
	Let $G$ be a strongly connected digraph with at least three vertices ($N \ge 3$). Then, the set $S^*_G$ spans~$\A^*$. Moreover, there exists a positive integer~$l$, with $l\le \d(G)$, such that $S^*_{G}\subseteq \ad^m(S_G)$ for any $m \ge l$.    
\end{theorem}

We provide below an example illustrating Theorem~\ref{thm:familycodistset}: 

\begin{example}{\em 
Consider a cycle digraph $G = (V, E)$ composed of three vertices: $V = \{v_1,v_2,v_3\}$ and $E = \{v_1v_2, v_2v_3, v_3v_1\}$. In this case, we have $S_G = \{A_{12}, A_{23}, A_{31}\}$ and 
$$
	S^*_G = \{A_{12} - A_{21}, A_{23} - A_{32}, A_{31} - A_{13}, 
	A_{13} - A_{12}, A_{21}- A_{23}, A_{32} - A_{31}\}.
$$ 
For the purpose of illustration, we write explicitly the matrices in the set $S^*_G$:  
    \begin{equation*}\label{eq:example123}
    \begin{array}{ll}
    A_{12} - A_{21} = 
    \begin{bmatrix}
    	-1 & 1 & 0\\
    	-1 & 1 & 0\\
    	0 & 0 & 0
    \end{bmatrix},	&
    A_{23} - A_{32} = 
    \begin{bmatrix}
    	0 & 0 & 0\\
    	0 & -1 & 1\\
    	0 & -1 & 1
    \end{bmatrix},	\vspace{3pt} \\
    A_{31} - A_{13} = 
    \begin{bmatrix}
    	1 & 0 & -1\\
    	0 & 0 & 0\\
    	1 & 0 & -1
    \end{bmatrix}, &	
    A_{13} - A_{12} = 
    \begin{bmatrix}
    	0 & -1 & 1\\
    	0 & 0 & 0\\
    	0 & 0 & 0
    \end{bmatrix},	\vspace{3pt} \\
    A_{21} - A_{23} = 
    \begin{bmatrix}
    	0 & 0 & 0\\
    	1 & 0 & -1\\
    	0 & 0 & 0
    \end{bmatrix},	&
    A_{32} - A_{31} = 
    \begin{bmatrix}
    	0 & 0 & 0\\
    	0 & 0 & 0\\
    	-1 & 1 & 0
    \end{bmatrix}.	
    \end{array}
\end{equation*}
The above matrices span~$\A^*$ (any five out of the six matrices form a basis).    
We show below that $S^*_G\subset \ad^m(S_G)$ for any $m \ge 3$ ($\d(G) = 3$).   
For convenience, we introduce an index set as follows: $$\mathcal{I}:=\{(1,2,3),(2,3,1),(3,1,2)\}.$$
All triplets in ${\cal I}$ can be obtained by a cyclic rotation of~$(1,2,3)$. 
Then, by computation, we have that  
$$
\begin{array}{l}
	\ad(S_G) = \pm \cup_{(i,j,k)\in \mathcal{I}} \{A_{ik} - A_{ij}\},\\
	\ad^2(S_G) =\pm \cup_{(i,j,k)\in \mathcal{I}} \{A_{ik} - A_{ij}, A_{ik} - A_{kj}\},
\end{array}
$$
and 
$$
\ad^3(S_G) =   \pm \cup_{(i,j,k)\in \mathcal{I}}\left \{A_{ij} - A_{ji}, A_{ik} - A_{ij}, 
 A_{ik} -A_{kj}, 2A_{ij} - A_{ik} - A_{ji} \right \}.
$$
It follows that $S^*_G\subset \ad^3(S_G)$. Moreover, since $\ad^2(S_G)\subset \ad^3(S_G)$, $\ad^m(S_G) \subset \ad^{m+1}(S_G)$ for all $m \ge 3$. We thus conclude that $S^*_G \subset \ad^m(S_G)$ for any $m \ge 3$. 
\hfill{\qed}
}  	
\end{example}

In the remainder of the section, we establish facts that are relevant to the proof of Theorem~\ref{thm:familycodistset}.

\subsection{Semi-codistinguished sets}
We consider here the adjoint action of the stochastic Lie algebra~$\A$ on its commutator ideal~$\A^*$, i.e., $\ad(A)(A^*)= [A, A^*]$ for any $A\in \A$ and $A^*\in \A^*$.   
We introduce the following definition adapted from~\cite{chen2018structure}:

\begin{definition}[Semi-codistinguished set]\label{def:codistinguished} 
	A subset $\{A^*_j\}^p_{j= 1}$ of $\A^*$ is {\bf semi-codistinguished} to a subset $\{A_i\}^m_{i = 1}$ of $\A$ if the following hold:
	\begin{enumerate}
	\item The set $\{A^*_j\}^p_{j = 1}$ spans $\A^*$.  
	\item For any $A^*_k$ in the set $\{A^*_j\}^p_{j = 1}$, there exist $A_i$, $A^*_j$, and a {\em nonzero} $\lambda$ such that 
		  \begin{equation}\label{eq:codist}
		  	[A_i, A^*_j] = \lambda A^*_k.
		  \end{equation}
	\end{enumerate}	  
\end{definition}
 
\begin{remark} {\em 
A stronger notion, termed {\em codistinguished set}, is introduced in~\cite{chen2018structure} (which was defined for arbitrary Lie algebras): A set $\{A^*_{j}\}^p_{j = 1}$ is codistinguished to $\{A_i\}^m_{i = 1}$ if it is semi-codistinguished, and moreover, for any $A_i$ and any $A^*_j$, there exist  $A^*_k$ and a constant $\lambda$ (which could be zero) such that~\eqref{eq:codist} holds. Existence of a codistinguished set for the case where $\g$ is semi-simple was addressed in~\cite{chen2017distinguished}. 
}  
\end{remark}

We establish in the subsection the following result: 

\begin{proposition}\label{prop:semi-codist}
	If $G$ is strongly connected, then $S^*_G$ is semi-codistinguished to~$S_G$. Moreover, for any matrix $A^*_k\in S^*_G$, there exist $A_i\in S_G$ and $A^*_j\in S^*_G$ such that the nonzero constant~$\lambda$ in~\eqref{eq:codist} takes value~$1$.    
\end{proposition}

Note that Prop.~\ref{prop:semi-codist} implies that if there exists an~$l$ such that $S^*_G\subset \ad^l(S_G)$, then $S^*_G \subset \ad^m(S_G)$ for any $m \ge l$. We prove Prop.~\ref{prop:semi-codist} below:

\begin{proof}[Proof of Proposition~\ref{prop:semi-codist}.]
	We prove the proposition by showing that the two items of Def.~\ref{def:codistinguished} are satisfied. 
	We first show that $S^*_G$ spans $\A^*$. Denote by~$\K$ the complete graph on~$N$ vertices. Correspondingly, we have that 
	$$
	S^*_\K = \{A_{jk} - A_{kj}, A_{ik} - A_{ij} \mid v_j \neq v_k\}. 
	$$ 
 	We note here the fact that the set $S^*_\K$ spans~$\A^*$. We omit a proof of the fact, but refer to~\cite{boukas2015structure} for details. The authors there provided a basis of~$\A^*$. The elements of the basis can be realized as integer combinations of the matrices  in $S^*_{\rm K}$.  It now suffices to show that each matrix in~$S^*_\K$ can be expressed as a linear combination of the matrices in~$S^*_G$. 
	
	We start by considering matrices of type $A_{ik} - A_{ij}\in S^*_\K$. Since~$G$ is strongly connected, there exists a path from~$v_j$ to~$v_k$. Denote such a path by $v_{i_0} v_{i_1}\cdots v_{i_l}$ where $v_{i_0} = v_j$ and $v_{i_l} = v_k$. By definition of $S^*_G$~\eqref{eq:defsG}, we have that 
	$$\{A_{i, i_p} - A_{i,i_{p-1}} \mid p = 1,\ldots, l \} \subset S^*_G.$$  
	One can thus express $A_{ik} - A_{ij}$ as follows:
	$$
	A_{ik} - A_{ij} = \sum^{l}_{p = 1} (A_{i, i_p} - A_{i,i_{p-1}}), 
	$$
	which is an integer combination of matrices in $S^*_G$.

	We next consider matrices of type $A_{jk} - A_{kj}\in S^*_\K$. We again let  $v_{i_0} \cdots v_{i_l}$ be the path from $v_j$ to $v_k$. Similarly, one can express $A_{jk} - A_{kj}$ as follows: 
	\begin{multline*}
	A_{jk} - A_{kj} = (A_{i_0i_l} - A_{i_0i_1}) + (A_{i_li_{l-1}} - A_{i_li_0}) \\ +  \sum^{l}_{p = 1}(A_{i_{p-1}i_p} - A_{i_pi_{p-1}})  + \sum^{l-1}_{p = 1}(A_{i_pi_{p-1}} - A_{i_pi_{p+1}}).
	\end{multline*}
	By~\eqref{eq:defsG}, 
	each term $(A_{i_{p-1}i_p} - A_{i_pi_{p-1}})$ for $p = 1, \ldots, l $ on the right-hand side of the above expression belongs to $S^*_G$. By the earlier arguments, any other term on the right-hand side can be expressed as an integer combination of matrices in $S^*_G$. 
	We have thus shown that every matrix in $S^*_\K$ can be expressed as an integer combination of matrices in $S^*_G$. Since $S^*_\K$ spans $\A^*$,  $S^*_G$ spans~$\A^*$ as well.

	Finally, we show that any matrix in $S^*_G$ can be obtained as a matrix commutator $[A, A^*]$ with $A\in S_G$ and $A^*\in S^*_G$. But this directly follows from the computation: For any $v_jv_k \in E$, we have 
	$$
	\begin{array}{l}
		{ [A_{jk}, A_{jk} - A_{kj}] }= A_{jk} - A_{kj}, \\ 
		{ [A_{jk}, A_{ik} - A_{ij}] }= A_{ik} - A_{ij}.
	\end{array}
	$$
	This completes the proof. 
\end{proof}

\begin{remark}\label{rmk:notcodist}{\em
	We note here that the set $S^*_G$ is in general {\em not} codistinguished to $S_G$. Consider, for example, the case where $G$ is the complete graph $\K$. If $1\le i\neq j \neq k\le N$, then by computation, we have that 
	$$
	\begin{array}{l}
		{ [A_{ji}, A_{ik} - A_{ij}] } =  (A_{jk} - A_{ji}) - (A_{ij} -A_{ji}),  \\
		{ [A_{ki}, A_{ik} - A_{ij}] } = (A_{ki} - A_{ik}) - (A_{ki} - A_{kj}), \\
		{ [A_{ik}, A_{ij} - A_{ji}] }  = (A_{ik} - A_{ij}) - (A_{ji} - A_{jk}).
	\end{array}
	$$
	The matrices on the right-hand sides of the above expressions do not belong to $S^*_\K$, but are expressed as integer combinations of the matrices in~$S^*_\K$.  }
\end{remark}

\subsection{Proof of Theorem~\ref{thm:familycodistset}.}  
We establish in the subsection Theorem~\ref{thm:familycodistset}. With Prop.~\ref{prop:semi-codist}, it remains to show that there exists a positive integer~$l$, with $l \le \d(G)$, such that $S^*_G \subseteq \ad^l(S_G)$. The result will be established after a sequence of lemmas. We start with the following fact:

\begin{lemma}\label{lem:negativethesame}
	For any $l\ge 1$, $\ad^l(S_G) = - \ad^l(S_G)$.
\end{lemma}

\begin{proof}
The proof can be carried out by induction on~$l$. For the base case~$l = 1$, we have that for any $A, A'\in S_G$, $[A, A'] = -[A', A]$. For the inductive step, we assume that the lemma holds for~$(l - 1)$ and prove for~$l$ (with $l \ge 2$). Consider any $[A, A']\in \ad^l(S_G)$ with $A\in S_G$ and $A'\in \ad^{l -1}(S_G)$. By the induction hypothesis, $-A'\in \ad^{l-1}(S_G)$. It follows that $[A, -A'] = - [A, A'] \in \ad^l(S_G)$.  	
\end{proof}

Recall that there are two different types of matrices in the set $S^*_G$, namely $A_{ik} - A_{ij}$ and $A_{jk} - A_{kj}$ (where $v_j v_k$ is an edge of~$G$). To show that every type of matrix can be obtained by an iterated matrix commutator of primary matrices in $S_G$, we need the following two lemmas (Lemmas~\ref{lem:type2matrix} and~\ref{lem:type1matrix}):

\begin{lemma}\label{lem:type2matrix}
	Let $v_{i_0}\cdots v_{i_l}$ be a path of length~$l$ in the digraph~$G$. Suppose that $l \ge 2$ and $v_{i_0} \neq v_{i_l}$; then,  
	$$
	A_{i_0 i_l } - A_{i_0 i_{l-1}} \in \ad^{l-1}(S_G). 
	$$
	See Fig.~\ref{fig:generator1} for an illustration.
\end{lemma}

\begin{proof}
	The proof will be carried out by induction. For the base case where $l = 2$, we have that $[A_{i_0i_1}, A_{i_1i_2}] = A_{i_0 i_2} - A_{i_0 i_1}  = \ad(S_G)$. 
	For the inductive step, we assume that the lemma holds for $(l - 1)$ with $l \ge 3$, and prove for~$l$. By the induction hypothesis, we have~$A_{i_0i_{l-1}} - A_{i_0i_{l-2}}\in \ad^{l-2}(S_G)$. By Lemma~\ref{lem:negativethesame}, $-(A_{i_0i_{l-1}} - A_{i_0i_{l-2}})\in \ad^{l-2}(S_G)$. It then follows that 
	$$
	[A_{i_{l-1}i_l}, A_{i_0i_{l-2}} - A_{i_0i_{l-1}}] = A_{i_0i_l} - A_{i_0i_{l-1}} \in \ad^{l-1}(S_G).$$ 
 This completes the proof. 
\end{proof}

\begin{figure}[h]
	\begin{center}
	\includegraphics[width = 0.6\textwidth]{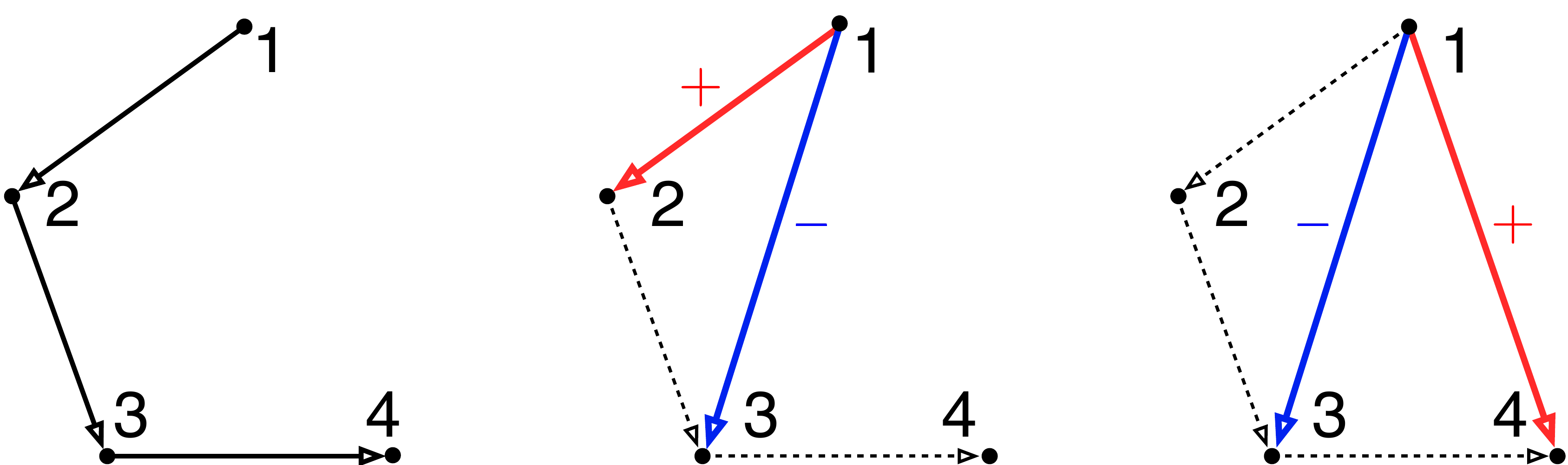}
	\caption{
	In the figure, $v_1v_2v_3v_4$ is a path of a digraph $G$. The set of primary matrices $\{A_{12}, A_{23}, A_{34}\}$ belongs to $S_G$. We show how to generate $A_{14} - A_{13}$ by an iterated matrix commutator of those primary matrices. In the first step, we have $[A_{23}, A_{12}] = A_{12} - A_{13}\in \ad(S_G)$.  
	The plus/minus sign of an edge $v_iv_j$ in the figure indicates the sign of the associated matrix $A_{ij}$ in the expression. 
	Next, we have $[A_{34}, [A_{23}, A_{12}]] = [A_{34}, A_{12} - A_{13}] = A_{14} - A_{13}\in \ad^2(S_G)$.    
	}\label{fig:generator1}
	\end{center}
\end{figure}

We next have the following fact: 

\begin{lemma}\label{lem:type1matrix}
	Let $v_{i_0}v_{i_1} \cdots v_{i_l}v_{i_0}$ be a cycle of length $(l+1)$ in the digraph~$G$. Suppose that $l  \ge 2$; then, 
	$$A_{i_0 i_1 } - A_{i_1 i_0} \in \ad^{l + 1} (S_G).$$
See Fig.~\ref{fig:generator2} for an illustration. 
\end{lemma}
 
\begin{proof}
	Consider the path $v_{i_2}\cdots v_{i_l}v_{i_0}v_{i_1}$ of length $l \ge 2$. By Lemma~\ref{lem:type2matrix}, we have 
	$A_{i_2 i_1 } - A_{i_2 i_0} \in \ad^{l - 1}(S_G)$. Next, by computation, we obtain
	$$
	[A_{i_1i_2}, A_{i_2 i_1 } - A_{i_2 i_0}] =  A_{i_2i_1} - A_{i_1i_0}  \in  \ad^l(S_G).
	$$ 
	Further, we have 
	$$
	[A_{i_0i_1}, A_{i_2i_1} - A_{i_1i_0} ] =  A_{i_0 i_1} - A_{i_1 i_0 } \in  \ad^{l + 1}(S_G).
	$$
	This completes the proof.
\end{proof}

\begin{figure}[h]
\begin{center}
	\includegraphics[width = 0.7\textwidth]{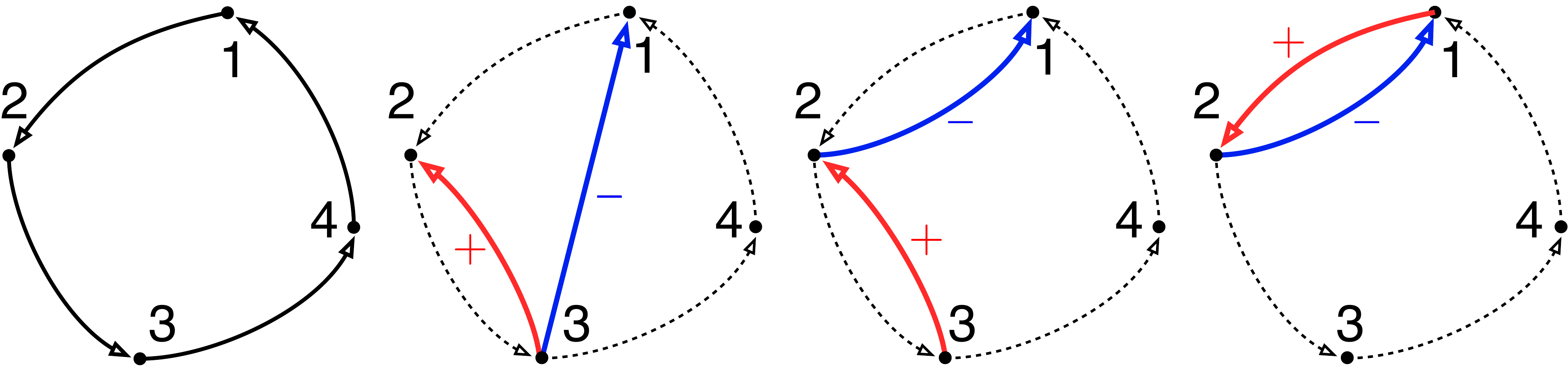}	
	\caption{
	In the figure, $v_1v_2v_3v_4v_1$ is a cycle of length~$4$. The set of primary matrices $\{A_{12}, A_{23}, A_{34}, A_{41}\}$ belongs to $S_G$.
	We show how to generate $A_{12} - A_{21}$ by an iterated matrix commutator of those primary matrices. By Lemma~\ref{lem:type2matrix},   we have $A_{32} - A_{31}\in \ad^2(S_G)$ via the path $v_3v_4v_1v_2$. 
	Next, we have $[A_{23}, A_{32} - A_{31}] = A_{32} - A_{21} \in \ad^3(S_G)$, and finally, $[A_{12}, A_{32} - A_{21}] = A_{12} - A_{21}\in \ad^4(S_G)$.   
	}\label{fig:generator2}
\end{center}
\end{figure}

With the lemmas above, we are now in a position to prove Theorem~\ref{thm:familycodistset}: 

\begin{proof}[Proof of Theorem~\ref{thm:familycodistset}.]
We show that $S^*_{G}\subseteq \ad^l(S_G)$ for some $l \le \d(G)$.       
First, consider the matrix $A_{jk} - A_{kj}$ with $v_j v_k\in E$. If $v_kv_j$  is an edge of $G$, then $A_{jk} - A_{kj} = [A_{kj}, A_{jk}]\in \ad(S_G)$. 
Now, suppose that $v_k v_j$ is not an edge of $G$; then, since $G$ is strongly connected, there exists a cycle $v_{i_0}v_{i_1}\cdots v_{i_{l'}}v_{i_0}$ with $v_{i_0} = v_j$ and $v_{i_1} = v_k$, and the integer~$l'$ satisfies $2 \le l' \le \d(G) - 1$.   
By Lemma~\ref{lem:type1matrix}, 
$A_{jk} - A_{kj} \in \ad^{l}(S_G)$ with $l :=(l' + 1) \le \d(G)$.

We next consider the matrix $A_{ik} - A_{ij}$ with $v_j v_k \in E$. If $v_iv_j$ is an edge of $G$, then $[A_{ij}, A_{jk}] = A_{ik} - A_{ij}\in \ad(S_G)$. 
Now, suppose that $v_iv_j$ is not an edge of $G$; then, there exists a path $v_{i_0} \cdots v_{i_{l}}$ with $v_{i_0} = v_i$ and $v_{i_{l}} = v_j$, and $l$ satisfies $2\le l \le \d(G)$. By Lemma~\ref{lem:type2matrix}, 
$A_{ij} - A_{i,i_{l-1}}\in  \ad^{l-1}(S_G)$. Then, by computation, we have
$$
[A_{jk}, A_{ij} - A_{i,i_{l-1}}] = -(A_{ik} - A_{ij}) \in \ad^{l}(S_G).
$$ 
The above holds regardless of whether $v_{i_{l-1}} = v_k$ or not. By Lemma~\ref{lem:negativethesame}, $A_{ik} - A_{ij}\in \ad^l(S_G)$.  

Finally, by Prop.~\ref{prop:semi-codist}, the set $S^*_G$ is semi-codistinguished to $S_G$. 
We thus conclude that $S^*_G\subset \ad^l(G)$ for some~$l\le \d(G)$ and $S^*_G \subset \ad^m(S_G)$ for all $m \ge l$.  
\end{proof}

\section{Analysis of ensemble formation system}\label{sec:proofofmainresult}

We investigate in the section controllability of an ensemble formation system and establish Theorem~\ref{thm:main}. For convenience, we
reproduce below the dynamics:
\begin{equation}\label{eq:rewrittensystem}
\dot X_\sigma(t) = \sum_{v_iv_j \in E}\sum^r_{s = 1} u_{ij, s}(t)\rho_s(\sigma)A_{ij}X_{\sigma}(t), \quad \sigma\in \Sigma.  
\end{equation}
As was mentioned at the end of Section~\ref{sec:Main}, the proof of Theorem~\ref{thm:main} relies on the use of Lie extension of~\eqref{eq:rewrittensinglesystem}. The technique of Lie extension has been widely used in the proof of approximate controllability of a control-affine system~\cite{sussmann1993lie,liu1997approximation} and in nonholonomic motion planning~\cite{murray1993steering}.  
We now review such a technique below.     

\subsection{Lie extended formation systems}\label{ssec:lieextendedsys}
To proceed, we consider an arbitrary control-affine system as follows: 
\begin{equation}\label{eq:singlecontroaffine}
\dot x(t) = \sum^m_{i = 1}u_i(t)f_i(x(t)),
\end{equation}
where the $u_i(t)$'s are control inputs and the $f_i(x)$'s are control vector fields. 
Let $F:= \{f_i\}^m_{i = 1}$, and $\mathfrak{L}(F)$ be the associated free Lie algebra generated by the set~$F$ where the $f_i$'s are treated as if they were free generators. 
Let ${\cal F}$ be a basis of $\mathfrak{L}(F)$, composed of Lie products of the~$f_i$'s. Recall that for a given~$k\ge 0$, ${\cal F}(k)$ is a subset of ${\cal F}$ composed of all formal Lie products of the $f_i$'s of depth~$k$. 
Then, the Lie extension of system~\eqref{eq:singlecontroaffine} gives rise to a family of control-affine systems as follows: Given a nonnegative integer~$k$, we have the following~$k$th order Lie extended system:   
$$
\dot x(t) = \sum^k_{l = 0}\sum_{f\in {\cal F}(l)} u_f(t) f(x(t)), 
$$ 
where each~$f$ is a formal Lie product of the~$f_i$'s.  Note that the control inputs $u_{f}(t)$ are independent of each other. We have the following well known fact (see~\cite{sussmann1993lie,liu1997approximation} and~\cite{agrachev2016ensemble}):

\begin{lemma}\label{lem:relationship}
System~\eqref{eq:singlecontroaffine} is approximately path-controllable if and only if any of its Lie extended system is.  	
\end{lemma}
 
We now apply the technique of Lie extension to the ensemble formation system~\eqref{eq:rewrittensystem}. First, note that for a given parameter $\sigma\in \Sigma$, the control vector fields of system-$\sigma$ are given by
	$\rho_s(\sigma)A_{ij}X_\sigma$ for  $s = 1,\ldots, r$ and  $v_iv_j \in E$. 
	The Lie bracket of any two of these vector fields is given by
	$$
	[\rho_s(\sigma)A_{ij}X_\sigma, \rho_{s'}(\sigma)A_{i'j'}X_\sigma] = \rho_s(\sigma) \rho_{s'}(\sigma) [A_{i'j'}, A_{ij}]X_\sigma.
	$$
	So, for example, the first order Lie extended system of~\eqref{eq:rewrittensystem} is given by the following ensemble system: For all $\sigma\in \Sigma$,  
	$$
	\dot X_\sigma(t) = \sum_{\alpha} u_{\alpha}(t)\rho_s(\sigma)A_{ij}X_{\sigma}(t)   
	+ \sum_{\beta} u_{\beta}(t)\rho_s(\sigma)\rho_{s'}(\sigma)[A_{ij}, A_{i'j'}] X_\sigma(t), 
	$$     
	where the two summations are over admissible multi-indices $\alpha = (s, v_iv_j)$ and $\beta = (s, s', v_iv_j, v_{i'}v_{j'})$. We make the statement precise below.  
	
	For a general $k$th order Lie extended system of~\eqref{eq:rewrittensystem}, we have the following:   
	First, recall that ${\cal P}(k)$ is the collection of monomials $\prod^r_{s = 1}\rho^{k_s}_s$ of degree~$k =\sum^r_{s = 1}k_s$. We also recall that $S_G$ is the collection of primary matrices $A_{ij}$ with $v_iv_j\in E$. We let $\mathfrak{L}(S_G)$ be the free Lie algebra generated by the set $S_G$ as if the primary matrices were free generators. Let ${\cal S}$ be a subset of $\mathfrak{L}(S_G)$, composed of all formal Lie products of the matrices in $S_G$. Decompose the set ${\cal S} = \sqcup_{k \ge 0}{\cal S}(k)$ where each ${\cal S}(k)$ is composed of formal Lie products of depth~$k$ in~${\cal S}$.  
	Then, the $k$th order Lie extended system can be expressed as follows: For all $\sigma\in \Sigma$,     
	\begin{equation}\label{eq:lieextension}
	\dot X_\sigma(t) = \sum^k_{l = 0}\sum_{\rp \in {\cal P}(l + 1)}\sum_{A\in {\cal S}(l)} u_{A,\rp}(t) \rp(\sigma) AX_\sigma(t), 
	\end{equation}
	where each $A$ on the right-hand side of the above expression is a formal Lie product of the $A_{ij}$'s in~$S_G$.

	To proceed, we recall that by Theorem~\ref{thm:familycodistset}, the set $S^*_G$ defined in~\eqref{eq:defsG} spans $\A^*$. For convenience, we let $$\gamma := \dim \A^*  = N(N - 1) - 1,$$ 
 	and let $\{A^*_i\}^\gamma_{i = 1}$ be any subset of $S^*_G$ such that the matrices $A^*_i$ form a basis of $\A^*$. 
	By the same theorem, we also have that $S^*_G \subset \ad^{l}(S_G)$ for any $l \ge \d(G)$.   
	Thus, there exists a subset of formal Lie products in $\ad^l(S_G)$ such that if one evaluates these formal Lie products (i.e., one computes the iterated matrix commutators), then the set of the resulting matrices contains~$S^*_G$ as a subset.    
 	Let $\{A^{[l]}_i\}^\gamma_{i= 1}$, for $l \ge \d(G)$, be such a subset out of $\ad^{l}(S_G)$: 
 	\begin{equation}\label{eq:keytotheresult}
 	A^{[l]}_i = A^*_i, \quad \forall i = 1,\ldots, \gamma.  
 	\end{equation}
 	Let ${\cal S}^*$ be a subset of ${\cal S}$ composed of the formal Lie products $A_i^{[l]}$ for all $i = 1,\ldots, \gamma$ and for all $l \ge \d(G)$, i.e., ${\cal S}^*:= \{A^{[l]}_i \mid 1\le i \le \gamma, \,\, l \ge \d(G)\}$.

 	Now, we fix a positive integer $k \ge \d(G)$ and consider the~$k$th order Lie extended system~\eqref{eq:lieextension}.  Let a control input $u_{A, \rp}(t)$ be identically zero if the formal Lie product~$A$  does not belong to the subset ${\cal S}^*$ defined above. Then, by~\eqref{eq:keytotheresult}, system~\eqref{eq:lieextension} can be reduced to the following: For all $\sigma\in \Sigma$, 
 	\begin{equation}\label{eq:thelastsystem}
 	\dot X_\sigma(t) = \sum^\gamma_{i = 1} 
 	\left [\sum^k_{l = \d(G)}\sum_{\rp \in {\cal P}(l + 1)} u_{i, \rp}(t) \rp(\sigma) \right ]
 	 A^*_iX_\sigma(t).
 	\end{equation}

 	We are now in a position to state the main result of the section.   Recall that the subset $Q\subset P$ is composed of all nondegenerate configurations. By Lemma~\ref{lem:resultforsetQ}, if $N > (n + 1)$, then $Q$ is path-connected, open, and dense in $P$. 
 	We now have the following fact: 

	\begin{proposition}\label{prop:controllabilityoflieextended}
		Let $\hat X_\Sigma[0,T]$ be a smooth trajectory with $\hat X_\sigma(t)\in Q$ for all $(t,\sigma)\in [0,T]\times \Sigma $. Suppose that the assumption of Theorem~\ref{thm:main} is satisfied; then, there exist a positive integer~$k \ge \d(G)$ and smooth control inputs $u_{i, \rp}:[0,T]\to \R$ for system~\eqref{eq:thelastsystem} such that for any $\sigma\in \Sigma$,  the trajectory $X_\sigma[0,T]$ generated by the system, with $X_\sigma(0)\in Q$  and $\|X_\sigma(0) - \hat X_\sigma(0)\| < \epsilon$, satisfies: 
		\begin{equation}\label{eq:lastcondition}
		\|X_\sigma(t) - \hat X_\sigma(t)\| < \epsilon, \quad \forall (t,\sigma)\in [0,T] \times \Sigma.
		\end{equation}
	\end{proposition}	
 
	Theorem~\ref{thm:main} then follows from Lemma~\ref{lem:relationship} and the above proposition. The remainder of the section is devoted to the proof of Prop.~\ref{prop:controllabilityoflieextended}. 
	A critical constituent of the proof is to show that for each $X\in Q$,  the linear span of $\{A^*_i X\}^\gamma_{i = 1}$ is the entire Euclidean space $\R^{N\times n}$ (which can be viewed as the tangent space of~$Q$ at any~$X$). 
	We prove such a fact in the next subsection.

\subsection{Linear span of control vector fields}\label{ssec:linearspanofcontrolvecfield}
We consider here an auxiliary {\em single} control-affine system associated with~\eqref{eq:thelastsystem}:   
$$
\dot X(t) = \sum^\gamma_{i = 1} u_i(t) A^*_i X(t),
$$
where the matrices $\{A^*_i\}^\gamma_{i = 1}$ were defined in the previous subsection; they were chosen out of the set $S^*_G$ and form a basis of $\A^*$. We illustrate the vector fields $A^*_iX$ in Fig.~\ref{fig:vectorfield}.     

 	\begin{figure}[h]
 	\begin{center}
 		\includegraphics[width = 0.65\textwidth]{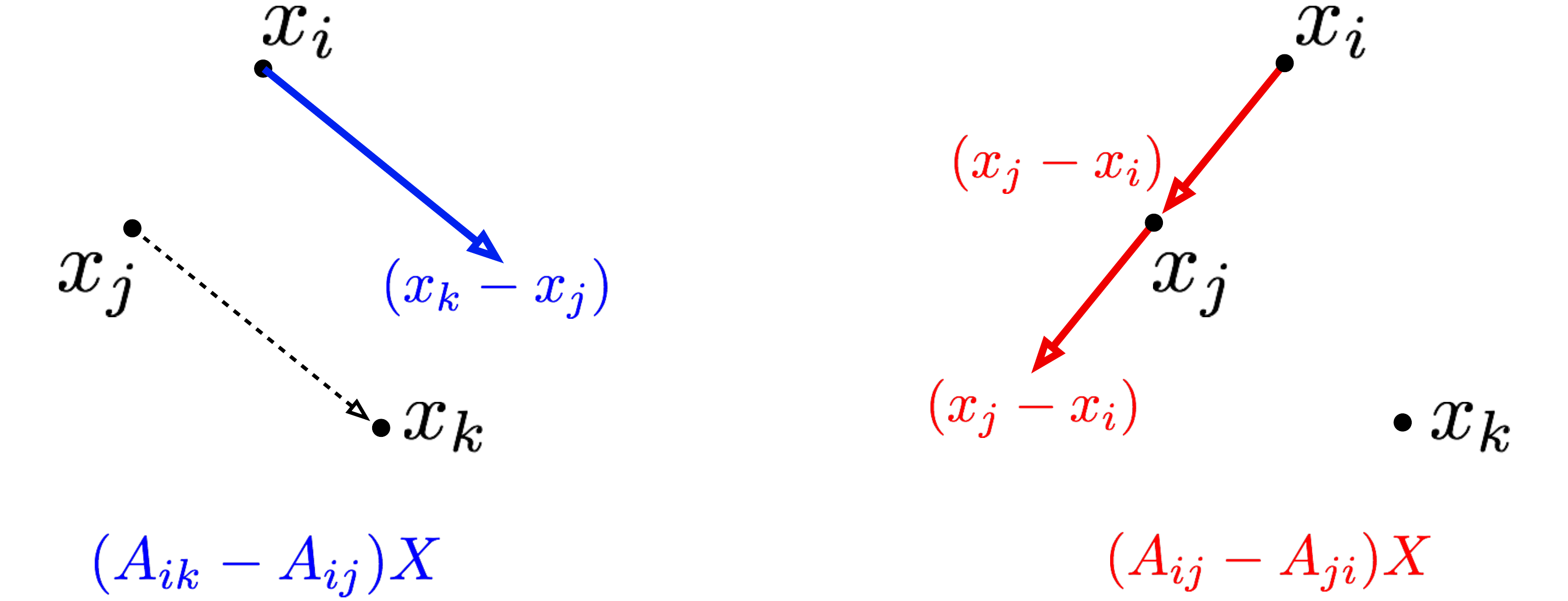}
 		\caption{We illustrate $A^*_i X$ by describing the infinitesimal motions of the~$N$ agents. Since the matrices $A^*_{i'}$'s are chosen out of $S^*_G$, either $A^*_{i'} = A_{ik} - A_{ij}$ or $A^*_{i'} = A_{ij} - A_{ji}$. If $A^*_{i'} = A_{ik} - A_{ij}$, then only agent $x_i$ has nonzero infinitesimal motion given by~$(x_k - x_j)$. If $A^*_{i'} = A_{ij} - A_{ji}$, then agents $x_i$ and $x_j$ have nonzero infinitesimal motions given by~$(x_j - x_i)$.}\label{fig:vectorfield}
 		\end{center}
 	\end{figure}

We show that the above system is {\em exactly} path-controllable over~$Q$ by proving the fact that if $X\in Q$, then $\{A^*_i X\}^\gamma_{i = 1}$ spans the entire Euclidean space~$\R^{N\times n}$($\approx T_XQ$).    
For a matrix $X\in Q$, we define a subspace $\bbL^*_{X}$ of $\R^{N\times n}$ as follows: 
$$
\bbL^*_X := \{AX\mid A\in \A^*\}. 
$$ 
We establish below the following fact: 

\begin{proposition}\label{prop:fullrankLX}
	If $ N > (n + 1)$ and $X\in Q$, then $\bbL^*_X  = \R^{N\times n}$.  
\end{proposition}

To prove the above proposition, we first recall a  fact established in~\cite{chen2017controllability}. Define a subspace of $\R^{N\times n}$ as follows:  
$$\bbL_X := \{AX\mid A\in \A\}.$$ 
Since $\A^*\subsetneq \A$, we have $\bbL^*_X \subseteq \bbL_X$.  
We established in~\cite{chen2017controllability} the following fact as a weaker version of Prop.~\ref{prop:fullrankLX}:  

\begin{lemma}\label{lem:simplecase}
	If $N > n$ and $X\in Q$, then $\bbL_X = \R^{N\times n}$.   
\end{lemma}

The proof of the lemma is built upon the construction of a set of $nN$ linearly independent vectors~$A_{ij} X$'s where each $A_{ij}$ is a primary matrix.  However, note that any primary matrix $A_{ij}$ has trace~$-1$, and hence does not belong to~$\A^*$. In the proof of Prop.~\ref{prop:fullrankLX}, we will construct a new set of matrices $A^*_i$'s in $\A^*$ so that the $A^*_iX$'s span $\R^{N\times n}$.

Also, by comparing Prop.~\ref{prop:fullrankLX} with Lemma~\ref{lem:simplecase}, we see that if $N > (n + 1)$ and $X\in Q$, then the equality $\bbL^*_X = \bbL_X = \R^{N\times n}$ holds. On the other hand, we note here that if $N = n + 1$ and $X\in Q$, then $\bbL^*_X \subsetneq \bbL_X$. To see this, first note  that 
	$$\dim \A^* = N(N - 1) - 1 = n(n + 1) - 1.$$ It then follows that 
	$$\dim \bbL^*_X \le \dim \A^* = n(n + 1) - 1 < n(n+1) = \dim \bbL_X.$$  
	In other words, the condition $N > (n + 1)$ is also {\em necessary} for the equality $\bbL^*_X = \bbL_X = \R^{N\times n}$ to hold. 	
With Lemma~\ref{lem:simplecase} as a preliminary result, we now prove Prop.~\ref{prop:fullrankLX}:

\begin{proof}[Proof of Prop.~\ref{prop:fullrankLX}.]
	Since the configuration $X\in Q$ is nondegenerate and $N \ge (n + 2)$, by Remark~\ref{rmk:nondegenerate} there exists a subset of $(N - 1)$ agents such that the subconfiguration formed by these $(N - 1)$ agents is nondegenerate in~$\R^n$. Without loss of generality, we assume that the subset is composed of the first $(N - 1)$ agents $\{x_i\}^{N - 1}_{i = 1}$. 
	
	We now introduce a subset $L^*_X$ of $\bbL^*_{X}$, and show that $L^*_X$ spans  $\R^{N\times n}$. The subset of interest is defined as follows:
	$$
	L^*_X:= \{ (A_{ij} - A_{Nk}) X \mid  
	1\le i\neq j \le N-1, 1\le k \le N - 1\}.
	$$ 
	Note that the leading $(N - 1)\times (N - 1)$ principal minor of $(A_{ij} - A_{Nk})$ is simply a primary matrix. Moreover, the collection of all such principal minors is the set $S_{\K'}$ where $\K'$ is the complete subgraph induced by the first $(N - 1)$ vertices $\{v_i\}^{N - 1}_{i = 1}$.  
	 
	Let $X':= 
	 [
	 x^\top_1; \cdots ;  x^\top_{N - 1}
	 ]
	 \in \R^{(N - 1)\times n}$ be the subconfiguration formed by the first $(N - 1)$ agents. Correspondingly, we partition a matrix $Y\in L^*_X$ as follows: 
	\begin{equation}\label{eq:submatrix}
	Y = 
	\begin{bmatrix}
	Y ' \\ 
	y^\top 	
	\end{bmatrix},
	\quad  Y' \in \R^{(N -1) \times n} \mbox{ and } y\in \R^{n}. 
	\end{equation}
	By the above arguments, we can write $Y' = A'X'$ for $A'\in S_{\K'}$. From the definition of $L^*_{X}$, the collection of all such $Y'$ is $\{A'X' \mid A'\in S_{\K'}\}$. 
	
	Denote by $\A'$ the stochastic Lie algebra composed of all $(N - 1)\times (N - 1)$ zero-row-sum matrices. Because $S_{\K'}$ spans $\A'$, the span of all the $Y'$ is given by $\bbL_{X'}:=\{A' X' \mid A'\in \A'\}$.  
	Furthermore, since $(N - 1)> n$,  by Lemma~\ref{lem:simplecase}, 
	$\dim \bbL_{X'} = n(N - 1)$, i.e., $\bbL_{X'} = \R^{(N - 1)\times n}$. By the above arguments, we know that there exists a subset $L'_X$ of $L^*_X$,  composed of 
	$n(N - 1)$ matrices $\{Y_i\}^{n(N - 1)}_{i = 1}$, such that the sub-matrices $\{Y'_i\}^{n(N - 1)}_{i= 1}$ (obtained by the partition~\eqref{eq:submatrix}) form a basis of $\bbL_{X'}$. We fix any such subset $L'_X$.

	It now suffices to find another~$n$ matrices in $\bbL^*_X$ such that they together with the $Y_i$'s span~$\R^{N\times n}$. Since the subconfiguration $X'$ formed by the first $(N - 1)$ agents is nondegenerate and $N > n + 1$, by Remark~\ref{rmk:nondegenerate} there exist $(n + 1)$ agents out of $\{x_i\}^{N - 1}_{i = 1}$ such that they form an $n$-simplex. Without loss of generality, we assume that the $(n + 1)$ agents are $\{x_i\}^{n + 1}_{i = 1}$. Then, we define a subset of $\bbL^*_X$ as follows:
	$$
	L''_X:= \{(A_{N,n+1} - A_{Nk})X \mid 1 \le k \le  n \}.  
	$$ 
	Note that each matrix in the above set is an integer combination of the matrices in $L^*_X$; indeed, we have
	$$
	(A_{N,n+1} - A_{Nk})X = (A_{1j} - A_{Nk}) X- (A_{1j} - A_{N,n+1})X,  
	$$ 
	for some  $j \in \{2,\ldots, N - 1\}$.  
	It should be clear that there are~$n$ matrices in~$L''_X$. 
	
	We show below that the union $L'_X\cup L''_X$ spans $\R^{N\times n}$. 
	First, by computation, we have 
	$$
	Z_k:=(A_{N, n + 1} - A_{Nk}) X = 
	 \begin{bmatrix}
	 	{\bf 0}_{(N -1) \times n} \\
	 	(x_{n + 1} - x_k)^\top 
	 \end{bmatrix}, 
	$$ 
	for all $k = 1,\ldots, n$.
	Since the agents $\{x_k\}^{n + 1}_{k = 1}$ form an $n$-simplex, by Def.~\ref{def:nondegenerate}, the~$n$ vectors $\{x_{n + 1} - x_k\}^{n}_{k = 1}$ span $\R^n$. Thus, $\{Z_k\}^n_{k = 1}$ are linearly independent. 
	
	Furthermore, by the construction, the matrices  $\{Z_k\}^n_{k = 1}$ are linearly independent of the matrices $\{Y_i\}^{n(N - 1)}_{i = 1}$; indeed, the sub-matrices $Y'_i$ of $Y_i$ form a basis of $\bbL_{X'}$ while the corresponding  sub-matrices of the $Z_k$'s are all zeros as shown in the above computation.  
	 So, there are $nN$ linearly independent matrices in the union $L'_X \cup L''_X$. We thus conclude that    
	$\span(L'_X \cup L''_X) = \span(L^*_{X}) = \R^{N\times n}$ for any $X\in Q$. 
\end{proof}

\subsection{Controllability of Lie extended formation systems}\label{ssec:controllabilityoflieextsys} 
	We prove here Prop.~\ref{prop:controllabilityoflieextended}. Recall that $\hat X_\Sigma[0,T]$ is the desired trajectory we want the Lie extended formation system~\eqref{eq:thelastsystem} to approximate. 
	Since the initial condition $X_\Sigma(0)$ may differ from $\hat X_\Sigma(0)$, we choose a smooth trajectory $\tilde X_\Sigma[0,T]$ such that $\tilde X_\Sigma(0) = X_\Sigma(0)$ and 
	\begin{equation}\label{eq:firststepapproximation}
	\|\tilde X_\sigma(t) - \hat X_\sigma(t)\| < \epsilon, \quad \forall (t, \sigma)\in [0,T]\times \Sigma.
	\end{equation}
	Because $X_\sigma(0)\in Q$ for all $\sigma\in \Sigma$, $\hat X_\sigma(t)\in Q$ for all $(t, \sigma)\in [0,T]\times \Sigma$, and $Q$ is open, we can choose $\tilde X_\Sigma[0,T]$ such that each $\tilde X_\sigma(t)$, for $(t, \sigma)\in [0,T]\times \Sigma$,  belongs to~$Q$ as well.  In the case where $X_\Sigma(0) = \hat X_\Sigma(0)$, we can simply let $\tilde X_\Sigma[0,T] = \hat X_\Sigma[0,T]$. 
	
	Next, we consider the time-derivative $\partial \tilde X_\sigma(t)/\partial t\in \R^{N\times n}$. Since $\tilde X_\sigma(t)$ belongs to $Q$, by Prop~\ref{prop:fullrankLX}, we have that  $$\span\{A^*_i \tilde X_\sigma(t)\}^\gamma_{i = 1} = \R^{N\times n}, \quad \forall(t,\sigma)\in [0,T]\times \Sigma.$$ 
	In particular, there exists a set of smooth functions $\{c_i(t, \sigma)\}^\gamma_{i = 1}$ in both~$t$ and~$\sigma$ such that 
	$$
	\frac{\partial }{\partial t} \tilde X_\sigma(t) = \sum^\gamma_{i = 1} c_i(t, \sigma) A^*_i\tilde X_\sigma(t).
	$$   
	Suppose, for the moment, that for any positive number~$\delta$, there exist a positive integer~$k$ for~$k\ge \d(G)$ and a set of smooth control inputs $u_{i, \rp}(t)$ for the $k$th order Lie extended formation system~\eqref{eq:thelastsystem} such that 
	\begin{equation}\label{eq:approximateci}
	\left |\sum^k_{l = \d(G)}\sum_{\rp\in {\cal P}(l + 1)} u_{i, \rp}(t)\rp(\sigma) -  c_i(t, \sigma) \right | < \delta, \quad
	\forall (t,\sigma)\in [0,T]\times \Sigma \, \mbox{and} \, \forall i = 1,\ldots, \gamma; 
	\end{equation} 
	then, it follows that for all~$\sigma\in \Sigma$, the trajectory $X_\sigma[0,T]$ generated by system~\eqref{eq:thelastsystem}, with $X_\sigma(0) = \tilde  X_\sigma(0)$, can be made arbitrarily close to~$\tilde X_\sigma[0,T]$.  
	
	Note that if the above holds, then  Prop.~\ref{prop:controllabilityoflieextended} will be established. To see this, first note that $[0,T]\times \Sigma$ is compact.  By~\eqref{eq:firststepapproximation}, there exists an $\epsilon'$, with $0 < \epsilon' < \epsilon$, such that if we replace $\epsilon$ with $\epsilon'$, then ~\eqref{eq:firststepapproximation} still holds. 
	Next, we let $\delta$ be chosen sufficiently small so that if~\eqref{eq:approximateci} holds, then 
	$$\|X_\sigma(t) - \tilde X_\sigma(t) \| < \epsilon - \epsilon', \quad \forall (t, \sigma)\in [0,T]\times \Sigma.$$ Combining the above arguments, together with the triangle inequality, we obtain that 
	\begin{multline*}
	\|X_\sigma(t) - \hat X_\sigma(t)\| \le \|X_\sigma(t) - \tilde X_\sigma(t)\| + \|\tilde X_\sigma(t) - \hat X_\sigma(t)\| \\ < (\epsilon - \epsilon') + \epsilon' = \epsilon, \quad \forall (t, \sigma)\in [0,T]\times \Sigma,  
	\end{multline*}      
	and hence~\eqref{eq:lastcondition} holds.

	It now remains to show that for any given $\delta > 0$, there exist a positive integer $k \ge \d(G)$ and smooth control inputs $u_{i, \rp}[0,T]$ such that~\eqref{eq:approximateci} holds. We establish the fact below.

	By assumption, $\{\rho_s\}^r_{s = 1}$ separates points and 
	contains an  everywhere nonzero function. Without loss of generality, we let $\rho_1$ be such a function. Note, in particular, that $\rho^{-1}_1$ is defined. 
	Because $\Sigma$ is compact, by the Stone-Weierstrass theorem, the subalgebra generated by $\{\rho_s\}^r_{s = 1}$ is dense in the algebra of all integrable functions defined on~$\Sigma$. 	 
	Thus, given any $\delta' > 0$ and any $i = 1,\ldots, \gamma$, there exist
	\begin{enumerate}
	\item[\em 1)] 	a nonnegative integer~$k'_i$, 
	\item[\em 2)] 	a subset of monomials $\{\rp'_{i_j}\}^{m_i}_{j = 1}$ out of~$\sqcup^{k'_i}_{l = 0}{\cal P}(l)$, and 
	\item[\em 3)] a set of smooth control inputs $\{u_{i_j}(t)\}^{m_i}_{j = 1}$, 
	\end{enumerate}  
such that the following holds: 
	\begin{equation}\label{eq:withindelta}
	\left | \sum^{m_i}_{j = 1}u_{i_j}(t)\rp'_{i_j}(\sigma)  - \rho^{-\d(G) -1}_1(\sigma) c_i(t,\sigma) \right |< \delta', 
	\quad\forall (t, \sigma) \in [0, T]\times \Sigma.
	\end{equation}
	Because~$\Sigma$ is compact and $\rho_1$ is everywhere nonzero, $$
	\kappa:= \max\left \{ \left |\rho^{-\d(G) - 1}_1(\sigma) \right | \mid \sigma\in \Sigma \right \}$$ 
	exists and is strictly positive. 
	Now, let $$\delta':= \delta/\kappa \quad  \mbox{and}  \quad \rp_{i_j}:= \rho^{\d(G) + 1}_1\rp'_{i_j}.$$   
	Then, the degree of each $\rp_{i_j}$ is at least $(\d(G) + 1)$. Moreover, by~\eqref{eq:withindelta}, we obtain that
	$$
	\left | \sum^{m_i}_{j = 1}u_{i_j}(t)\rp_{i_j}(\sigma)  - c_i(t,\sigma) \right |< \kappa \delta' = \delta, 
	\quad\forall (t,\sigma)\in [0,T]\times \Sigma \, \mbox{and} \, \forall i = 1,\ldots, \gamma.
	$$
	We have thus established~\eqref{eq:approximateci}.  \hfill{\qed}

\begin{remark}{\em  
	The arguments used in the proof of Prop.~\ref{prop:controllabilityoflieextended}, when combined with the averaging techniques established in~\cite{sussmann1993lie,liu1997approximation}, could be used to design an algorithm for generating a set of control inputs $u_{ij, s}[0,T]$ for steering the original ensemble formation system~\eqref{eq:rewrittensystem} to approximate a given trajectory $\hat X_\Sigma[0,T]$. In a nutshell, the algorithm is composed of two steps: {\em (i)} Identify the order~$k$ of Lie extension and computes the control inputs $u_{i, \rp}(t)$ for the Lie extended system so that for any $\sigma\in \Sigma$, the trajectory generated by~\eqref{eq:thelastsystem} is within a certain error tolerance of the desired trajectory $\hat X_\sigma[0,T]$; {\em (ii)} Apply the averaging technique~\cite{sussmann1993lie,liu1997approximation}, which takes as input the $k$th order Lie extended system (with the $u_{i, \rp}(t)$'s computed above) and yields an appropriate set of control inputs $u_{ij, s}(t)$ that steer~\eqref{eq:rewrittensystem} to approximate the trajectory generated by the Lie extended system. We defer the analysis to another occasion.            
	     }
\end{remark}

\section{Conclusions}

We investigated in the paper a continuum ensemble of multi-agent formation systems~\eqref{eq:startensemble} with their information flows described by a common digraph~$G$. Such an ensemble control model can be viewed as a prototype for the study of design and control of other general ensembles of networked control systems. We established in the paper a sufficient condition (Theorem~\ref{thm:main}) for the ensemble formation system~\eqref{eq:startensemble} to be approximate path-controllable, i.e., for every individual formation system to be simultaneously approximately path-controllable. In particular, the theorem related ensemble controllability to the (common) information flow topology of every individual formation system.  

To establish the controllability result, we investigated the stochastic Lie algebra~$\A$ and computed iterated matrix commutators of the primary matrices~$A_{ij}$ for $A_{ij}\in S_G$.    
We introduced the notion of semi-codistinguished set (Def.~\ref{def:codistinguished}) and showed that for  every strongly connected digraph $G$, there exists a set $S^*_G$ (defined in~\eqref{eq:defsG}) semi-codistinguished to~$S_G$ with respect to the adjoint representation of $\A$ on its commutator ideal~$\A^*$. Moreover, we showed that such a set $S^*_G$ can be generated by iterated matrix commutators of primary matrices of any given depth that is greater than or equal to the diameter of~$G$. The above analysis of stochastic Lie algebra was instrumental in evaluating iterated Lie brackets of control vector fields of the ensemble formation system. We verified in Section~\ref{sec:proofofmainresult} the ensemble version of the Lie algebraic rank condition. 

Future work may focus on extending the controllability result to the case where the information flow digraphs $G_\sigma(t)$, for $\sigma\in \Sigma$, are time-varying and heterogeneous.
We will also aim to address the observability of an ensemble formation system, which is about the ability of using a finite number of output measurements to estimate the state of every individual formation system in the ensemble.

\bibliographystyle{IEEEtran}
\bibliography{ensembleformationcontrol,ensemblecontrol,nsfcareer}

\section*{Appendix}  
We prove here Lemma~\ref{lem:perfectliealgebra}. 
Let~$U$ be an $N\times N$ nonsingular matrix such that $U \1 = e_N$. Define a subspace $\tilde \A^*$ of $\R^{N\times N}$ by $\tilde \A^*:= U\A^* U^{-1}$, i.e., 
$
\tilde \A^*:= \{UAU^{-1} \mid A\in \A^*\} 
$.     
It should be clear that $\tilde \A^*$ is isomorphic to $\A^*$ (as a matrix Lie algebra). We show below that $\tilde \A^*$ is perfect. Recall that $\A^*$ is defined by the conditions that $A \1 = 0$ and $\tr(A) = 0$ for any $A\in \A^*$. Translating these conditions to the set $\tilde \A^*$, we obtain that 
$$
\tilde \A^* = \{\tilde A \in \R^{N\times N} \mid \tilde A e_N = 0 \mbox{ and } \tr(\tilde A) = 0\}.
$$ 
Note, in particular, that the right column of $\tilde A'\in \tilde \A^*$ is zero.  
Decompose a matrix $\tilde A\in \tilde \A^*$ into $2\times 2$ blocks: 
$
\tilde A = 
[
\tilde A', 0; 
b^\top, 0	
]
$
where $\tilde A' \in \R^{(N - 1) \times (N - 1)}$ and $b\in \R^{N - 1}$. Since $\tr(\tilde A') = \tr(\tilde A) = 0$, the collection of all such submatrices $\tilde A'$ is then the special linear Lie algebra $\mathfrak{sl}_{N - 1}(\R)$, which is known to be simple (note that $N \ge 3$). 

It follows that the Lie algebra $\tilde \A^*$ can be expressed as a semidirect product $\tilde \A^* = \mathfrak{sl}_{N -1}(\R) \ltimes \R^{N-1}$. Specifically, if we represent a matrix $\tilde A \in \tilde \A^*$ by a pair $(A', b^\top)$, then the Lie bracket of $(\tilde A'_1, b_1^\top)$ and $(\tilde A'_2, b_2^\top)$ is given by
$$
[(\tilde A'_1, b_1^\top), (\tilde A'_2, b_2^\top)] = ([\tilde A'_1, \tilde A'_2], b_1^\top \tilde A'_2 - b_2^\top  \tilde A'_1 )
$$  
Recall that the standard representation of $\mathfrak{sl}_{N -1}(\R)$ on $\R^{N - 1}$ is given by $(\tilde A', b) \mapsto \tilde A' b $. The representation is known to be irreducible, i.e., there does not exist a nonzero, proper subspace $V\in \R^{N-1}$ such that $\sl_{N-1}(\R) V \subseteq V$. 
This, in particular, implies that $\mathfrak{sl}_{N-1}(\R) \R^{N - 1} =  \R^{N-1}$. It then follows that the Lie algebra $\tilde \A^* = \sl_{N-1}(\R)\ltimes \R^{N-1}$ is perfect. 

Let $\tilde \A^* = \tilde \A^*_\mathfrak{l} \oplus \tilde \A^*_\mathfrak{r}$ be the Levi decomposition of $\tilde \A^*$ where $\tilde \A^*_\mathfrak{l}$ is semi-simple and $\tilde \A^*_\mathfrak{r}$ is the radical of $\tilde \A^*$.   
It should be clear from the above arguments that  
$$
\begin{array}{l}
\tilde \A^*_\mathfrak{l} = \left \{
\begin{bmatrix}
\tilde A' & 0 \\
0 & 0 	
\end{bmatrix}
\mid \tr\left (\tilde A' \right ) = 0
\right \}\approx \sl_{N-1}(\R), \vspace{3pt}\\
\tilde \A^*_\mathfrak{r} = \left \{
\begin{bmatrix}
0 & 0 \\
b^\top & 0 	
\end{bmatrix}
\mid b\in \R^{N - 1}
\right \};
\end{array}
$$  
indeed, we have that $[\tilde \A^*_\mathfrak{r}, \tilde \A^*_\mathfrak{r}] = 0$. 
Correspondingly, $\A^*_\mathfrak{l}$ and $\A^*_\mathfrak{r}$ can be obtained by the similarity transformation: $\A^*_\mathfrak{l}= U^{-1} \tilde A^*_\mathfrak{l} U$ and $\A^*_\mathfrak{r}= U^{-1} \tilde A^*_\mathfrak{r} U$. By computation, we have  
$$
\A^*_\mathfrak{l} = \{A\in \A^* \mid A^\top \1 = 0\} \quad \mbox{and} \quad \A^*_\mathfrak{r} = \{\1 v^\top \mid v^\top \1 = 0\}. 
$$
Note that the above result does not depend on a particular choice of~$U$ as long as $U$ is nonsingular  and $U \1$ is linearly proportional to $e_N$.  \hfill{\qed}

\end{document}